\documentclass[12pt,preprint]{imsart} 

\RequirePackage[OT1]{fontenc}
\RequirePackage{amsthm,amsmath}
\RequirePackage[numbers]{natbib}
\RequirePackage[colorlinks,citecolor=blue,urlcolor=blue]{hyperref}

\usepackage{graphicx,psfrag,epsf}
\usepackage{enumerate}
\usepackage{natbib}
\usepackage{mathtools}
\usepackage{booktabs}
\usepackage[english]{babel} 
\usepackage{microtype} 
\usepackage{amsmath,amsfonts,amssymb}
\usepackage{chngcntr}
\usepackage[toc,page]{appendix}
\usepackage{ amssymb }
\usepackage{array}
\usepackage{booktabs} 
\usepackage{natbib}
\usepackage{array,booktabs,arydshln,xcolor}
\usepackage{rotating}
\usepackage{soul}
\usepackage{threeparttable}
\usepackage{tablefootnote}
\usepackage{color}
\usepackage{tabularx}
\usepackage[font = small,labelfont=bf,textfont=it]{caption} 
\usepackage{footnote}
\usepackage{algorithm}
\usepackage[noend]{algpseudocode}
\usepackage{footnote}
\usepackage{multirow}
\usepackage{diagbox}

\makeatletter
\newcommand{\distas}[1]{\mathbin{\overset{#1}{\kern\z@\sim}}}%
\newcommand{\bm}[1]{\mathbf{#1}}
\newsavebox{\mybox}\newsavebox{\mysim}
\newcommand{\distras}[1]{%
  \savebox{\mybox}{\hbox{\kern3pt$\scriptstyle#1$\kern3pt}}%
  \savebox{\mysim}{\hbox{$\sim$}}%
  \mathbin{\overset{#1}{\kern\z@\resizebox{\wd\mybox}{\ht\mysim}{$\sim$}}}%
}
\newtheorem{theorem}{Theorem}

\newtheorem{definition}{Definition}

\newtheorem{lemma}[theorem]{Lemma}

\newtheorem{proposition}{Proposition}

\newcolumntype{C}[1]{>{\centering\let\newline\\\arraybackslash\hspace{0pt}}m{#1}}

\newcommand{\ull}[1]{\underline{#1}}
\newcommand{\ulcrd}[1]{\ull{\ull{\crd{\textbf{#1}}}}}
\newcommand{\ulcbl}[1]{\ull{\cbl{#1}}}
\newcommand{\be}{\begin{equation}}
\newcommand{\ee}{\end{equation}}
\newcommand{\bi}{\begin{itemize}}
\newcommand{\ei}{\end{itemize}}
\newcommand{\ben}{\begin{enumerate}}
\newcommand{\een}{\end{enumerate}}
\newcommand{\stb}{\State $\bullet$ \;}
\newcommand{\crd}[1]{{\color{red}{#1}}}
\newcommand{\cbl}[1]{{\color{blue}{\textbf{#1}}}}

\newcolumntype{"}{@{\hskip\tabcolsep\vrule width 1pt\hskip\tabcolsep}}
\allowdisplaybreaks

\makeatother

\let\oldbibliography\thebibliography
\renewcommand{\thebibliography}[1]{\oldbibliography{#1}
\setlength{\itemsep}{0pt}} 

\begin{document}

\begin{frontmatter}
\title{Robust designs for Gaussian process emulation of computer experiments}
\runtitle{Robust designs for computer experiments}

\begin{aug}
\author{\fnms{Simon} \snm{Mak}\ead[label=e1]{sm769@duke.edu}}\footnote{Department of Statistical Science, Duke University} and
\author{\fnms{V. Roshan} \snm{Joseph}\ead[label=e2]{roshan@gatech.edu}}\footnote{School of Industrial \& Systems Engineering, Georgia Institute of Technology}

\runauthor{S. Mak and V. R. Joseph}

\end{aug}

\begin{abstract}
\; We study in this paper two classes of experimental designs, support points and projected support points, which can provide robust and effective emulation of computer experiments with Gaussian processes. These designs have two important properties that are appealing for surrogate modeling of computer experiments. First, the proposed designs are robust: they enjoy good emulation performance over a wide class of smooth and rugged response surfaces. Second, they can be efficiently generated for large designs in high dimensions using difference-of-convex programming. In this work, we present a theoretical framework that investigates the above properties, then demonstrate their effectiveness for Gaussian process emulation in a suite of numerical experiments.
\end{abstract}



\end{frontmatter}

\section{Introduction}

With advances in scientific computing, computer experiments are quickly replacing physical experiments in modern scientific and engineering applications. Computer experiments employ sophisticated computer code to reliably simulate complex physical phenomena, with successful applications in rocket design \citep{bernardini2020large,mak2018efficient,chang2021reduced}, sustainable aviation \citep{miller2024diverse,narayanan2024misfire} and particle physics \citep{putschke2019jetscape,li2023additive,ji2022multi}. Despite such advances, computer experiments face an inherent challenge: the computer code for a single input setting can take days or weeks to run, which makes its exploration over a desired design space highly challenging. The \textit{design} of a computer experiment, i.e., the selection of which inputs to conduct simulations, is a crucial step for accurate emulation of the simulated response surface over the design space. This paper studies two class of designs, called \textit{support points} and \textit{projected support points}, which can provide {robust} and {effective} emulation of computer experiments.

Existing work on computer experiment design can be grouped into two categories: (a) {model-specific} designs, which depend on a given stochastic model for the response surface, and (b) {distance-based} designs, which optimize some criterion quantifying the distance between a design and its design space. The first category includes the integrated mean-squared error designs \citep{Sea1989a, Pea2017}, while the latter includes the so-called \textit{minimax} and \textit{maximin} designs \citep{Jea1990}. There is recent renewed interest in studying the theoretical properties and construction algorithms for the latter two flavors of design (e.g., \cite{Tan2013, MJ2017a, He2017} for minimax, and \cite{MM1995, DP2010} for maximin). The designs explored in this paper fall within the category of distance-based designs as well.

The notion of distance-based point sets has also garnered some interest in numerical integration. Of particular interest are the \textit{support points} (SPs) in \cite{MJ2017b}, which serve as representative points (rep-points) for a general distribution $F$. These points enjoy several important theoretical properties \citep{MJ2017b}: they converge in distribution to $F$, and also provide an improved integration rate to Monte Carlo for a large class of integrands. This is further extended in \cite{MJ2017c}, who introduced the \textit{projected support points} (PSPs) as a modification to SPs when the sample space of $F$ is high-dimensional. PSPs can be shown to be representative of not only the full distribution $F$, but also of its marginal distributions. As such, PSPs can provide good integration performance for high-dimensional functions with low-dimensional structure.

However, an ideal point set for \textit{representation} or \textit{integration} may not necessarily be good for the \textit{emulation} of response surfaces. One prominent example of this is the uniform design in \cite{Fan1980}, which minimizes the Kolmogorov-Smirnov statistic \citep{Kol1933} -- the worst-case discrepancy between the uniform hypercube $F = \mathcal{U}[0,1]^p$ and the empirical distribution of the design. While such designs enjoy an ideal representation of the uniform distribution $F$, they can have poor emulation performance \citep{Wie1991,Sea2013}. Conversely, when \textit{polynomials} are used as the basis-of-choice for emulation, it is well-known that Chebyshev nodes provide excellent interpolation performance \citep{Tre2013}, but such nodes converge to the arc-sine (and not the uniform) distribution. Similarly, \cite{WZ1997} and \cite{Puk2006} show that optimal designs for large-degree polynomial regression also converge to the same arc-sine distribution. \cite{DP2010} makes use of this arc-sine distribution to transform existing Latin hypercube designs for computer experiments.

We present in this paper a different but \textit{complementary} view to the above body of work -- we show that the support points in \cite{MJ2017b}, which provide good representation and integration on $F = \mathcal{U}[0,1]^p$, also offer excellent theoretical properties and practical performance for Gaussian process (GP) emulation. The key is in employing a \textit{different} measure of uniformity called the \textit{energy distance} \citep{SR2004}. In contrast to the usual Kolmogorov-Smirnov statistic, the energy distance uses a \textit{distance-based} kernel for quantifying deviances from uniformity for the design. This distance-based kernel yields three appealing design properties for SPs and PSPs, which we discuss in the paper. First and foremost, we show that these designs are \textit{robust} for Gaussian process modeling, in that they enjoy good emulation performance for both smooth functions (e.g., polynomials) and more rugged functions. We demonstrate this robustness property both theoretically via an information-theoretic argument, and empirically via simulation studies. Second, we derive important insights connecting SPs and PSPs with existing \textit{distance-based} designs, and demonstrate via such insights why the proposed designs can provide better performance. Lastly, using difference-of-convex optimization, we show that SPs and PSPs can be efficiently generated for large designs in high dimensions.

The paper is structured as follows. Section 2 introduces SPs and PSPs, and Section 3 argues why such designs are robust for Gaussian process emulation. Section 4 shows that SPs are a compromise between minimax and maximin designs, with greater weight on the former. Section 5 presents further insights connecting SPs and PSPs with existing designs. Section 6 details an improved algorithm for efficiently generating SPs and PSPs. Sections 7 and 8 explore the emulation performance of these designs in a suite of numerical experiments. Section 9 concludes with directions for future work.

\section{Preliminaries}
\label{sec:prelim}
We first provide a brief review of support points and projected support points, following \cite{MJ2017c,MJ2017b}.

\subsection{Support points}
\label{sec:sp}
Consider first the \textit{energy distance} between two distribution functions (d.f.s) $F$ and $G$, both defined on the sample space $\mathcal{X} \subseteq \mathbb{R}^p$:

\begin{definition}[Energy distance; \cite{SR2004}]
Let $\bm{X},\bm{X}' \;  \distas{i.i.d.} G$ and $\bm{Y}, \bm{Y}' \distas{i.i.d.} F$, where $F$ and $G$ are d.f.s on $\mathcal{X} \subseteq \mathbb{R}^p$ with $\mathbb{E}\|\bm{X}\|_2 < \infty$ and $\mathbb{E}\|\bm{Y}\|_2 < \infty$. The \textup{energy distance} between $F$ and $G$ is defined as:
\begin{equation}
\mathcal{E}(F,G) \equiv 2 \mathbb{E}\|\bm{X}-\bm{Y}\|_2 - \mathbb{E}\|\bm{Y} - \bm{Y}'\|_2 - \mathbb{E}\|\bm{X} - \bm{X'}\|_2.
\label{eq:energy}
\end{equation}
\label{def:energy}
\end{definition}
\vspace{-0.7cm}
\noindent An important property of the energy distance is that $\mathcal{E}(F,G)$ is non-negative, and attains a minimum value of 0 if and only if $F$ and $G$ are the same distribution. This is known as the \textit{metric property} in \cite{SZ2013}, and is important for ensuring that SPs (and its variants) provide a good representation of $F$.

In addition to having an appealing interpretation as a statistical potential measure \citep{SZ2013}, the energy distance has been widely used as an efficient goodness-of-fit test of a sample $\{\bm{x}_i\}_{i=1}^n$ to a pre-specified distribution $F$ \citep{SR2004}. In this light, the \textit{support points} of $F$ \citep{MJ2017b} are defined as the point set yielding the best goodness-of-fit to $F$:

\begin{definition}
[Support points; \cite{MJ2017b}]
For a given distribution $F$ on $\mathcal{X} \subseteq \mathbb{R}^p$ with finite mean, the \textup{support points (SPs)} for $F$ are defined as:
\begin{equation}
\underset{\bm{x}_1, \cdots, \bm{x}_n}{\textup{Argmin}} \; \mathcal{E}(F,F_n) = \underset{\bm{x}_1, \cdots, \bm{x}_n}{\textup{Argmin}} \; \left\{\frac{2}{n} \sum_{i=1}^n \mathbb{E}\|\bm{x}_i - \bm{Y}\|_2 - \frac{1}{n^2} \sum_{i=1}^n \sum_{j=1}^n \|\bm{x}_i - \bm{x}_j\|_2 \right\},
\label{eq:supp}
\end{equation}
where $\bm{Y} \sim F$, and $F_n$ is the empirical distribution function (e.d.f.) of the point set $\{\bm{x}_i\}_{i=1}^n \subseteq \mathcal{X}$.
\label{def:supp}
\end{definition}
\noindent One can show (Theorem 2 in \cite{MJ2017b}) that SPs converge in distribution to $F$ as $n \rightarrow \infty$, i.e., when the design grows large, so these points are indeed \textit{representative} of the desired distribution $F$.

From a design perspective, the formulation in \eqref{eq:supp} has the following charming interpretation. By \textit{minimizing} the first term $\sum_{i=1}^n \mathbb{E}\|\bm{x}_i - \bm{Y}\|_2$, one minimizes the expected Euclidean distance of each point to $F$, thereby ensuring all design points are close to the desired distribution. Likewise, by \textit{maximizing} the second term $\sum_{i=1}^n \sum_{j=1}^n \|\bm{x}_i - \bm{x}_j\|_2$, one maximizes the sum of pairwise distance between design points, thereby ensuring points are pushed apart over the design space $\mathcal{X}$. For the uniform hypercube $F = \mathcal{U}[0,1]^p$, these two objectives are quite similar to the goals for the distance-based minimax and maximin designs;  we explore this in greater detail in Section \ref{sec:minmax}.

\subsection{Projected support points}
\label{sec:psp}

While support points provide good representative points for the full distribution $F$, it may not provide a good representation for the marginal distributions of $F$. To this end, \cite{MJ2017c} proposed the following \textit{projected support points}:
\begin{definition}[Projected support points; \cite{MJ2017c}]
Define the generalized Gaussian kernel:
\begin{equation}
\gamma_{\boldsymbol{\theta}}(\bm{x},\bm{y}) = \exp\left\{ -\sum_{\emptyset \neq \bm{u} \subseteq [p]} \theta_{\bm{u}} \|\bm{x}_{\bm{u}} - \bm{y}_{\bm{u}}\|_2^2\right\},
\label{eq:kern}
\end{equation}
where $[p] = \{1, \cdots, p\}$. For a d.f. $F$ on $\mathcal{X} \subseteq \mathbb{R}^p$, the \textup{$\pi$-projected support points} for $F$ are defined as:
\begin{equation}
\underset{\bm{x}_1, \cdots, \bm{x}_n}{\textup{Argmin}} \; \left\{- \frac{2}{n} \sum_{i=1}^n \mathbb{E}_{\bm{Y},\boldsymbol{\theta} \sim \pi}\left[\gamma_{\boldsymbol{\theta}}(\bm{x}_i,\bm{Y})\right] + \frac{1}{n^2} \sum_{i=1}^n \sum_{j=1}^n \mathbb{E}_{\boldsymbol{\theta} \sim \pi}\left[\gamma_{\boldsymbol{\theta}}(\bm{x}_i,\bm{x}_j)\right] \right\},
\label{eq:psp}
\end{equation}
where $\bm{Y} \sim F$ and $\pi$ is a prior on the scale parameters $\boldsymbol{\theta}=(\theta_{\bm{u}})_{\emptyset \neq \bm{u} \subseteq [p]}$.
\label{def:psp}
\end{definition}
\noindent From a Quasi-Monte Carlo (QMC) point-of-view, PSPs minimize the kernel-based discrepancy \citep{Hic1998} under kernel $\gamma_{\boldsymbol{\theta}}$, subject to a prior distribution on scale parameters $\boldsymbol{\theta}$. The choice of the generalized Gaussian kernel in \eqref{eq:psp} will become clear in Section \ref{sec:maxpro}, when PSPs are compared with existing designs.

Comparing the formulation in \eqref{eq:psp} to that in \eqref{eq:supp}, we see that PSPs enjoy a similar interpretation to SPs from a design perspective -- design points are \textit{pulled} towards the desired distribution $F$, but also \textit{push} away from other design points. The key difference between SPs and PSPs is the latter employs the generalized Gaussian kernel in \eqref{eq:kern} to quantify similarities on both the \textit{full} design space $[0,1]^p$, as well as on its \textit{projected} subspaces $[0,1]^{\bm{u}}$, $\bm{u} \subseteq [p]$. In particular, \cite{MJ2017c} showed that a larger value of $\theta_{\bm{u}}$ (a) reflects greater importance of subspace $\bm{u}$, and (b) encourages better representativeness of the point set for the marginal distribution indexed by $\bm{u}$. The goal, then, is to judiciously choose the prior distribution $\pi$ so that only several coefficients in $\boldsymbol{\theta}$ are large with high probability. This then encourages the resulting PSPs in \eqref{eq:psp} to have good fit on both $F$ and its marginal distributions.

To this end, \cite{MJ2017c} proposed a prior specification on $\pi$ which imposes the three fundamental experimental design principles of effect hierarchy, heredity and sparsity \citep{BH1961, HW1992, WH2011}. In words, effect hierarchy presumes that lower-order effects are more likely active than higher-order ones in a response surface; effect heredity (in the strong sense) states that a higher-order effect is active only when all its lower-order components are active as well; and effect sparsity assumes the surface is comprised of a small number of effects. Together, these principles are fundamental to experimental design, because they allow for meaningful analysis using a limited amount of experimental data (which are typically expensive to collect). By imposing such principles into the generalized kernel \eqref{eq:kern}, the resulting PSPs should provide excellent designs for computer experiments, because such designs incorporate these expected response surface properties \textit{within} the design construction itself. This is indeed the case, and we explore this in greater depth in Section \ref{sec:con}.

\section{Support points as robust designs for GP modeling}
\label{sec:robust}
We first show why support points are robust and effective computer experiment designs when the underlying response surface follows a Gaussian process. A brief review of GP modeling is provided below for context. Unless otherwise stated, we assume from now on that $F = \mathcal{U}[0,1]^p$ is the uniform hypercube.

\subsection{Gaussian process modeling}
\label{sec:gp}
Let $f: [0,1]^p \rightarrow \mathbb{R}$ be the underlying response surface of a computer experiment. Suppose the vector of function responses $\bm{f} \equiv [f(\bm{x}_i)]_{i=1}^n$ has been collected at input settings $\mathcal{D} = \{\bm{x}_i\}_{i=1}^n$. A popular approach for estimating $f(\cdot)$ at an unobserved input setting $\bm{x}_{new}$ is to assume $f$ follows the Gaussian process $f(\cdot) \sim GP\{\mu, \sigma^2 r_{\boldsymbol{\theta}}(\cdot,\cdot)\}$. Here, $\mu$ and $\sigma^2$ are the process mean and variance, and $r_{\boldsymbol{\theta}}(\cdot,\cdot)$ is a correlation function with scale parameters $\boldsymbol{\theta} \in \mathbb{R}^p$.

Under such a model, the conditional expectation of $f(\bm{x}_{new})$ becomes \citep{Sea2013}:
\begin{equation}
\hat{f}(\bm{x}_{new}) \equiv \mathbb{E}[f(\bm{x}_{new})|\bm{f}] = \mu + \bm{r}_{new}^T\bm{R}_n^{-1}(\bm{f}-\mu \cdot \bm{1}),
\label{eq:pred}
\end{equation}
where $\bm{r}_{new} = [r_{\boldsymbol{\theta}}(\bm{x}_i,\bm{x}_{new})]_{i=1}^n$ and $\bm{R}_n = {[r_{\boldsymbol{\theta}}(\bm{x}_i,\bm{x}_j)]_{i=1}^n}_{j=1}^n$. Using this conditional expectation as the predictor for $f(\bm{x}_{new})$, the \textit{root-mean-squared error} (RMSE) of such a predictor can be written as the closed-form expression:
\begin{equation}
RMSE_{\boldsymbol{\theta}}(\bm{x}_{new}; \mathcal{D}) \equiv \sqrt{\text{Var}\{\hat{f}(\bm{x}_{new})\}} = \sigma \sqrt{1 - \bm{r}_{new}^T\bm{R}_n^{-1}\bm{r}_{new}}.
\label{eq:rmse}
\end{equation}
\noindent The fact that both prediction and uncertainty quantification can be expressed in closed-form makes the Gaussian process model the de-facto approach for modeling computer experiments \citep{Sea2013}.


\subsection{Robustness of support points}
One key challenge in designing experiments for GP modeling is that the correlation function $r_{\boldsymbol{\theta}}(\cdot,\cdot)$, along with its parameters $\boldsymbol{\theta}$, typically need to be specified a priori, even when little is known on the response surface \citep{Sea2013}. We outline below an argument for why SPs can serve as \textit{robust} designs for GP modeling, in that it provides near-optimal performance over a wide range of response surfaces, without the need for specifying model parameters a priori. 

Consider first the following decomposition of the response surface $f(\bm{x})$:
\begin{equation}
f(\bm{x}) = \mathcal{P}_{\mathcal{G}} f + \mathcal{P}_{\mathcal{G}}^{\bot} f (\bm{x}),
\label{eq:mean}
\end{equation}
where $\mathcal{G}$ is the space of constant functions on $[0,1]^p$, $\mathcal{P}_{\mathcal{G}} f \equiv \int_{[0,1]^p} f(\bm{x}) \; d\bm{x}$ is the mean of $f$ when projected onto $\mathcal{G}$, and $\mathcal{P}_{\mathcal{G}}^{\bot} f \equiv f - \mathcal{P}_{\mathcal{G}} f$ is the perpendicular component of $f$ representing its variation around its mean. The following proposition provides a useful property of the variation term $\mathcal{P}_{\mathcal{G}}^{\bot} f (\bm{x})$:
\begin{lemma}
Let $\mathcal{P}_{\mathcal{G}} f$ and $\mathcal{P}_{\mathcal{G}}^{\bot} f$ be as defined above. If $f \sim GP\{\mu, \sigma^2 r_{\boldsymbol{\theta}}(\cdot,\cdot)\}$, then $\mathcal{P}_{\mathcal{G}}^{\bot} f \sim GP\{ 0, \sigma^2 r_{\boldsymbol{\theta},\mathcal{G}}(\cdot,\cdot) \}$, where, with $\bm{Y}, \bm{Y}' \; \distas{i.i.d.} \mathcal{U}[0,1]^p$:
\begin{equation}
r_{\boldsymbol{\theta},\mathcal{G}}(\bm{x},\bm{y}) = r_{\boldsymbol{\theta}}(\bm{x},\bm{y}) - \mathbb{E}\{ r_{\boldsymbol{\theta}}(\bm{x},\bm{Y}) \} - \mathbb{E}\{r_{\boldsymbol{\theta}}(\bm{Y},\bm{y}) \}
+ \mathbb{E}\{ r_{\boldsymbol{\theta}}(\bm{Y},\bm{Y}')\}.
\label{eq:projkern}
\end{equation}
\label{lem:projkern}
\end{lemma}
\begin{proof}
This proof follows from a similar result in Section 3.1 of \cite{Tuo2017}. Suppose $f$ follows the Gaussian process $GP\{\mu,\sigma^2 r_{\boldsymbol{\theta}}(\cdot,\cdot)\}$. Since a Gaussian-distributed vector is still Gaussian-distributed after projection, it follows that the projected process $\mathcal{P}_{\mathcal{G}}^{\bot} f$ is also a Gaussian process. The mean of such a process is $\mathbb{E}\{ \mathcal{P}_{\mathcal{G}}^{\bot} f (\bm{x})\} = \mathbb{E}\left\{ f(\bm{x}) - \int_{[0,1]^p} f(\bm{z})\;d\bm{z}\right\} = \mathbb{E}\{f(\bm{x})\} - \int_{[0,1]^p} \mathbb{E}\{  f(\bm{z}) \} \; d\bm{z} = \mu - \mu = 0$, where the third-last step follows from the bounded convergence theorem (BCT), since $[0,1]^p$ is bounded. Similarly, its covariance function can be derived as:
\small
\begin{align*}
\textup{Cov}\{\mathcal{P}_{\mathcal{G}}^{\bot} f(\bm{x}), \mathcal{P}_{\mathcal{G}}^{\bot} f(\bm{y})\} &= \mathbb{E}\left\{ \left( f(\bm{x}) - \int_{[0,1]^p} f(\bm{z})\;d\bm{z} \right) \left( f(\bm{y}) - \int_{[0,1]^p} f(\bm{z}) \; d\bm{z} \right) \right\}\notag\\
&= \mathbb{E}\left\{ f(\bm{x}) f(\bm{y}) \right\} - \int_{[0,1]^p} \mathbb{E}\{f(\bm{x}) f(\bm{z}) \} \; d\bm{z}\notag\\
& \quad  - \int_{[0,1]^p} \mathbb{E}\{f(\bm{y}) f(\bm{z}) \} \; d\bm{z} + \int_{[0,1]^p} \int_{[0,1]^p} \mathbb{E}\{f(\bm{z})f(\bm{z}')\} \; d\bm{z} d\bm{z}' \\
&= \text{Cov}\left\{ f(\bm{x}), f(\bm{y}) \right\} - \int_{[0,1]^p} \text{Cov}\{f(\bm{x}),f(\bm{z}) \} \; d\bm{z}\notag\\
& \quad  - \int_{[0,1]^p} \text{Cov}\{f(\bm{y}), f(\bm{z}) \} \; d\bm{z} + \int_{[0,1]^p} \int_{[0,1]^p} \text{Cov}\{f(\bm{z}), f(\bm{z}')\} \; d\bm{z} d\bm{z}' \\
& \hspace{2.5cm} \text{(since $\mathbb{E}\{f(\bm{x})f(\bm{y})\} = \text{Cov}\{f(\bm{x}), f(\bm{y})\} + \mu^2$)}\\
&= \sigma^2 \left[ r_{\boldsymbol{\theta}}(\bm{x},\bm{y}) - \mathbb{E}\{ r_{\boldsymbol{\theta}}(\bm{x},\bm{Y}) \} - \mathbb{E}\{r_{\boldsymbol{\theta}}(\bm{Y},\bm{y}) \}
+ \mathbb{E}\{ r_{\boldsymbol{\theta}}(\bm{Y},\bm{Y}')\}\right], \notag
\end{align*}
\normalsize
where the third-last step follows from BCT. This completes the proof.
\end{proof}
\noindent In other words, when $f$ follows a GP with correlation $r_{\boldsymbol{\theta}}$, its projection $\mathcal{P}_{\mathcal{G}}^{\bot} f$ also follows a GP with so-called \textit{projected kernel} $r_{\boldsymbol{\theta},\mathcal{G}}$ (see \cite{PJ2016,Tuo2017} for details).

With this in hand, we now show that under general conditions on $f$, SPs can provide optimal estimation of both the \textit{mean} term $\mathcal{P}_{\mathcal{G}} f$ and the \textit{variation} term $\mathcal{P}_{\mathcal{G}}^{\bot} f (\bm{x})$. Consider the first goal of estimating the \textit{mean} term $\mathcal{P}_{\mathcal{G}} f = \int_{[0,1]^p} f(\bm{x}) \; d\bm{x}$, which can be viewed as a numerical integration problem with $f$ being the desired integrand. The following proposition, adapted from Theorem 4 of \cite{MJ2017b}, shows that SPs are optimal for estimating the integral $\mathcal{P}_{\mathcal{G}} f$ under general conditions on $f$:
\begin{proposition}[Estimation of $\mathcal{P}_{\mathcal{G}} f$; \cite{MJ2017b}]
Let $f: [0,1]^p \rightarrow \mathbb{R}$, and suppose $f \in W_{\lceil (p+1)/2 \rceil,2}$, where $W_{k,2}$, $k > 0$, is the \textup{Sobolev function space} on $[0,1]^p$:
\small
\begin{equation}
W_{k,2} = \Big\{ g \in \mathcal{L}_2([0,1]^p) \; : \; D^{\boldsymbol{\alpha}} g \in \mathcal{L}_2([0,1]^p), \; \forall |\boldsymbol{\alpha}| \leq k \Big\}, \;  D^{\boldsymbol{\alpha}} g \equiv \frac{\partial^{|\boldsymbol{\alpha}|} g}{\partial x_1^{\alpha_1} \cdots \partial x_p^{\alpha_p}},
\label{eq:sobolev}
\end{equation}
\normalsize
and $\mathcal{L}_2([0,1]^p)$ is the space of $\mathcal{L}_2$-integrable functions on $[0,1]^p$. Further, let $I(f;F,F_n) \equiv \left|\mathcal{P}_{\mathcal{G}} f - \int_{[0,1]^p} f(\bm{x}) \; dF_n(\bm{x})\right|$ be the integration error of $f$ under design $\mathcal{D} = \{\bm{x}_i\}_{i=1}^n$. Then:
\begin{equation}
I(f;F,F_n) \leq V(f) \sqrt{\mathcal{E}(F,F_n)},
\label{eq:interr}
\end{equation}
where $V(f)$ is a norm of the integrand $f$ (see \cite{MJ2017b}).
\label{thm:interr}
\end{proposition}

This proposition is a variant of the so-called Koksma-Hlawka inequality (see, e.g., \cite{Dea2013}), connecting the integration error $I(f;F,F_n)$ to the energy distance $\mathcal{E}(F,F_n)$ via an upper bound. Such a bound can be shown to be tight (i.e., there exists some function $f$ in the Sobolev space $W_{\lceil (p+1)/2 \rceil,2}$ for which \eqref{eq:interr} holds with equality), since it is obtained using the Cauchy-Schwarz inequality in an appropriately set reproducing kernel Hilbert space. SPs, by minimizing the energy distance $\mathcal{E}(F,F_n)$ in this tight upper bound, can therefore uniformly reduce the estimation error for $\mathcal{P}_{\mathcal{G}} f$ over all $f \in W_{\lceil (p+1)/2 \rceil,2}$, i.e., all response surfaces whose $\lceil (p+1)/2 \rceil$-th order differentials have finite $\mathcal{L}_2$-norm. 



Consider now the second goal of estimating the \textit{variation} term $\mathcal{P}_{\mathcal{G}}^{\bot} f (\bm{x})$; we employ below an information-theoretic argument to show why SPs also provide optimal estimation of such a term. Recall that the \textit{maximum entropy design} \citep{SW1987}, which maximizes the determinant of the correlation matrix $\bm{R}_n$, also maximizes the \textit{Shannon information} \citep{Sha1948} gain on the response surface $f$. Using an application of Gershgorin's circle theorem \citep{Ger1931}, one can derive the following approximate formulation for a maximum entropy design (see Section 5 of \cite{Jea2015} for details):
\begin{equation}
\underset{\bm{x}_1, \cdots, \bm{x}_n}{\textup{Argmax}} \; \det(\bm{R}_n) \approx \underset{\bm{x}_1, \cdots, \bm{x}_n}{\textup{Argmax}} \; \sum_{i=1}^n \left\{ \bm{r}_{\boldsymbol{\theta}}(\bm{x}_i,\bm{x}_i) - \sum_{j = 1, j \neq i}^n \bm{r}_{\boldsymbol{\theta}}(\bm{x}_i,\bm{x}_j) \right\}.
\label{eq:entapprox}
\end{equation}
Assuming the kernel $r_{\boldsymbol{\theta}}$ satisfies some bounded condition $\underline{r} \equiv \inf_{\bm{x}\in[0,1]^p}r_{\boldsymbol{\theta}}(\bm{x},\bm{x})$, the objective in \eqref{eq:entapprox} can then be lower bounded as:
\begin{equation}
\sum_{i=1}^n \left\{ 2\bm{r}_{\boldsymbol{\theta}}(\bm{x}_i,\bm{x}_i) - \sum_{j = 1}^n \bm{r}_{\boldsymbol{\theta}}(\bm{x}_i,\bm{x}_j) \right\} \geq 2n\underline{r} - \sum_{i=1}^n \sum_{j=1}^n \bm{r}_{\boldsymbol{\theta}}(\bm{x}_i,\bm{x}_j).
\end{equation}
Using this lower bound as a surrogate objective for the right-hand side of \eqref{eq:entapprox}, the maximum entropy designs in \eqref{eq:entapprox} can be well-approximated by designs which \textit{minimize} the sum of all entries in $\bm{R}_n$, i.e.:
\begin{equation}
\underset{\bm{x}_1, \cdots, \bm{x}_n}{\textup{Argmax}} \; \det(\bm{R}_n) \approx  \underset{\bm{x}_1, \cdots, \bm{x}_n}{\textup{Argmin}} \; \sum_{i=1}^n \sum_{j = 1}^n \bm{r}_{\boldsymbol{\theta}}(\bm{x}_i,\bm{x}_j).
\label{eq:entapprox2}
\end{equation}

Since the interest lies in estimating the \textit{variation} term $\mathcal{P}_{\mathcal{G}}^{\bot} f$ (which, by Lemma \ref{lem:projkern}, follows a zero-mean GP with correlation kernel $\bm{r}_{\boldsymbol{\theta}, \mathcal{G}}$), it makes sense to replace the original kernel $\bm{r}_{\boldsymbol{\theta}}$ in \eqref{eq:entapprox2} with the \textit{projected} kernel $\bm{r}_{\boldsymbol{\theta}, \mathcal{G}}$. The following theorem shows that the entry-wise sum in \eqref{eq:entapprox2} using $\bm{r}_{\boldsymbol{\theta}, \mathcal{G}}$ is proportional to the energy distance $\mathcal{E}(F,F_n)$:
\begin{theorem}[Information on $\mathcal{P}_{\mathcal{G}}^{\bot} f$]
Suppose $r_{\boldsymbol{\theta}}(\bm{x},\bm{y}) \propto \|\bm{x}\|_2 + \|\bm{y}\|_2 - \|\bm{x}-\bm{y}\|_2$. For any design $\mathcal{D} = \{\bm{x}_i\}_{i=1}^n$, we have:
\begin{equation}
\sum_{i=1}^n \sum_{j=i}^n r_{\boldsymbol{\theta},\mathcal{G}}(\bm{x}_i, \bm{x}_j) \propto \mathcal{E}(F,F_n).
\label{eq:sum}
\end{equation}
\label{thm:sum}
\end{theorem}
\vspace{-1.0cm}
\begin{proof}
Suppose $r_{\boldsymbol{\theta}}(\bm{x},\bm{y}) \propto \|\bm{x}\|_2 + \|\bm{y}\|_2 - \|\bm{x}-\bm{y}\|_2$, and take an arbitrary design $\mathcal{D} = \{\bm{x}_i\}_{i=1}^n \subseteq [0,1]^p$, with $F_n$ its empirical distribution function. Letting $F = \mathcal{U}[0,1]^p$, it follows that:
\small
\begin{align*}
\sum_{i=1}^n \sum_{j=i}^n r_{\boldsymbol{\theta},\mathcal{G}}(\bm{x}_i, \bm{x}_j) &= \sum_{i=1}^n \sum_{j=1}^n \left[ r_{\boldsymbol{\theta}}(\bm{x}_i,\bm{x}_j) - \mathbb{E}\{ r_{\boldsymbol{\theta}}(\bm{x}_i,\bm{Y}) \} - \mathbb{E}\{r_{\boldsymbol{\theta}}(\bm{Y},\bm{x}_j) \} + \mathbb{E}\{ r_{\boldsymbol{\theta}}(\bm{Y},\bm{Y}')\} \right] \notag \\
& = \text{\mbox{ \footnotesize $n^2 \Bigg[ \int_{[0,1]^p} \int_{[0,1]^p} r_{\boldsymbol{\theta}}(\bm{x},\bm{y}) \; dF_n(\bm{x}) dF_n(\bm{y}) - \int_{[0,1]^p} \int_{[0,1]^p} r_{\boldsymbol{\theta}}(\bm{x},\bm{y}) \; dF_n(\bm{x})dF(\bm{y}) $}}  \notag\\
& \quad \text{\mbox{\footnotesize $- \int_{[0,1]^p} \int_{[0,1]^p} r_{\boldsymbol{\theta}}(\bm{x},\bm{y}) \; dF(\bm{x}) dF_n(\bm{y}) + \int_{[0,1]^p} \int_{[0,1]^p} r_{\boldsymbol{\theta}}(\bm{x},\bm{y}) \; dF(\bm{x})dF(\bm{y}) \Bigg]$}} \notag\\
&= n^2\int_{[0,1]^p} \int_{[0,1]^p} r_{\boldsymbol{\theta}}(\bm{x},\bm{y}) \; d[F-F_n](\bm{x}) d[F-F_n](\bm{y}) \notag\\
&\propto \int_{[0,1]^p} \int_{[0,1]^p} \left( \|\bm{x}\|_2 + \|\bm{y}\|_2 - \|\bm{x}-\bm{y}\|_2 \right) \; d[F-F_n](\bm{x}) d[F-F_n](\bm{y}) \notag\\
&= \int_{[0,1]^p} \int_{[0,1]^p} - \|\bm{x}-\bm{y}\|_2 \; d[F-F_n](\bm{x}) d[F-F_n](\bm{y}) \notag\\
&= \text{\mbox{\footnotesize $2\int_{[0,1]^p} \int_{[0,1]^p} \|\bm{x}-\bm{y}\|_2 \; dF(\bm{x}) dF_n(\bm{y}) - \int_{[0,1]^p} \int_{[0,1]^p} \|\bm{x}-\bm{y}\|_2 \; dF_n(\bm{x})dF_n(\bm{y}) $}} \notag\\
&  \quad \text{\mbox{\footnotesize $- \int_{[0,1]^p} \int_{[0,1]^p} \|\bm{x}-\bm{y}\|_2 \; dF(\bm{x}) dF(\bm{y})$}}\\
&= \mathcal{E}(F,F_n), \notag
\end{align*}
\normalsize
where the third-last step follows because $\int_{[0,1]^p} \|\bm{x}\|_2 \; d[F-F_n](\bm{y}) $ = $\int_{[0,1]^p} \|\bm{y}\|_2 \; d[F-F_n](\bm{x}) = 0$. This completes the proof.
\end{proof}

Theorem \ref{thm:sum} can be interpreted as follows. It is known (see \cite{ZA2014}) that the correlation assumption $r_{\boldsymbol{\theta}}(\bm{x},\bm{y}) \propto \|\bm{x}\|_2 + \|\bm{y}\|_2 - \|\bm{x}-\bm{y}\|_2$ corresponds to the expected $1/2$-H\"{o}lder continuity condition $\mathbb{E}[\{f(\bm{x}) - f(\bm{y})\}^2] = \mathcal{O}(\|\bm{x}-\bm{y}\|_2)$ for the underlying GP sample paths. This condition is quite general and intuitive, because one would anticipate the expected \textit{variation} of the response surface to be proportional to the \textit{distance} between input settings. Under such an assumption, Theorem \ref{thm:sum} and \eqref{eq:entapprox2} show that a design with \textit{smaller} energy distance to $\mathcal{U}[0,1]^p$ can provide \textit{more} information on the variation term $\mathcal{P}_{\mathcal{G}}^{\bot} f (\bm{x})$. By minimizing $\mathcal{E}(F,F_n)$, SPs can yield a near-maximal information gain on $\mathcal{P}_{\mathcal{G}}^{\bot} f (\bm{x})$, which translates to good predictions of the variation of $f$ around its mean.

\begin{figure}
\small
\centering
\begin{tabular}{c " c " c }
\specialrule{2pt}{1pt}{1pt}
\textbf{Term} & \textbf{Assumptions on $f$} & \textbf{Optimality of SPs}\\
\specialrule{2pt}{1pt}{1pt}
Mean term & $f \in W_{\lceil (p+1)/2 \rceil,2}$ & Optimizes error upper bound\\
$\mathcal{P}_{\mathcal{G}} f$ & ($\mathcal{L}_2$-integrable $\lceil (p+1)/2 \rceil$-th & via Koksma-Hlawka \\
 & differentials) & (Theorem \ref{thm:interr})\\
\hline
Variation term & $\mathbb{E}[\{f(\bm{x}) - f(\bm{y})\}^2] = \mathcal{O}(\|\bm{x}-\bm{y}\|_2)$ & Near-maximal information\\
$\mathcal{P}_{\mathcal{G}}^{\bot} f(\bm{x})$ & (Expected $1/2$-H\"{o}lder continuity) & gain from design\\
& & (Theorem \ref{thm:sum}) \\
\specialrule{2pt}{1pt}{1pt}
\end{tabular}
\captionof{table}{Assumptions and optimality of SPs for predicting $\mathcal{P}_{\mathcal{G}} f$ and $\mathcal{P}_{\mathcal{G}}^{\bot} f(\bm{x})$.}
\label{tbl:robust}
\normalsize
\end{figure}

Table \ref{tbl:robust} summarizes the above results for (a) the assumptions on the response surface $f$, and (b) the optimality offered by SPs under such assumptions. We see that this class of response surfaces for SPs is indeed general; all that is assumed is the existence of integrable differentials and an expected H\"{o}lder continuity condition. Because of this generality, SPs can be viewed as \textit{robust} computer experiment designs, yielding good GP emulation performance over a large class of functions. This is in contrast to the model-specific integrated mean-squared error designs \citep{Sea1989a}, which require specification of both the correlation $r_{\boldsymbol{\theta}}(\cdot,\cdot)$ and its parameters $\boldsymbol{\theta}$, and are therefore not robust for modeling different response surfaces. Of course, one way of circumventing this is to assume a prior structure on $\boldsymbol{\theta}$ and adopt a Bayesian form of the design criterion (see, e.g., \citep{Pea2017}), but such an approach is more computationally demanding, and is also sensitive to the choice of prior specification. Compared to these methods, SPs provide the desired robustness property, and can also be efficiently generated for large designs in high-dimensions via difference-of-convex programming (see Section \ref{sec:alg}).

To illustrate this robustness, consider the emulation of two simple 1-d test functions: the scaled Runge function $f_1(x) = 1/(1+100(x-0.5)^2)$, and the modified step function $f_2(x) = 0.25 + 0.5(x-a)/(b-a) \cdot \mathbf{1}\{a \leq x < b\} + 0.5 \cdot \mathbf{1}\{x \geq b\}$, with $a = 0.4$ and $b = 0.6$ (this modification is made to satisfy the H\"{o}lder condition in Table \ref{fig:robust}). Figure \ref{fig:robust} shows the GP fits of these functions using the Gaussian correlation $\exp\{- 10 (x - y)^2\}$, with $n=7$ SPs and Chebyshev nodes. For the smooth Runge function, SPs provide comparable (but slightly worse) performance to Chebyshev nodes, which is unsurprising since the latter is well-known for being optimal interpolant points for the Runge function \citep{Tre2013}. For the non-smooth step function, however, SPs enjoy considerably improved predictions to Chebyshev nodes. This nicely demonstrates the \textit{robustness} of SPs for modeling both smooth and non-smooth functions in 1-d; we show later in Sections \ref{sec:eff} and \ref{sec:app} that the same robustness also holds for higher-dimensional functions.


\begin{figure}[t]
\centering
\includegraphics[width=0.75\textwidth]{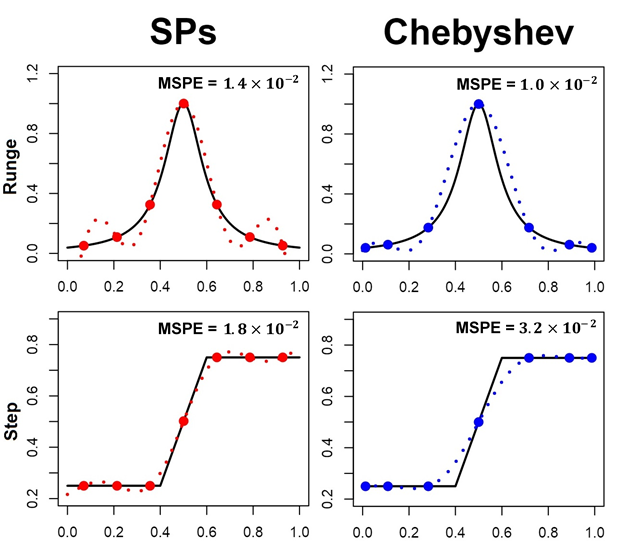}
\caption{GP fits of the scaled Runge function and the modified step function using $n = 7$ SPs and Chebyshev nodes. Solid lines indicate true functions, dotted lines indicate predictions, dots indicate observed function points, and MSPE indicates mean-squared prediction error.}
\label{fig:robust}
\end{figure}

\subsection{Support points and fractional Brownian motion modeling}
\label{sec:fbm}
The correlation function in Theorem \ref{thm:sum} also has a useful connection to a special class of GPs called the \textit{fractional Brownian motion} (fBm), which we define below:
\begin{definition}[Fractional Brownian motion; \cite{Man1985}]
Let $q \in (0,2)$, and let $Z(\bm{x})$ be a stochastic process with the following properties: (a) $Z(\bm{x})$ is a Gaussian process, (b) $Z(\bm{0}) = 0$ almost surely, (c) $\mathbb{E}[Z(\bm{x}) - Z(\bm{y})] = 0$, and (d) $\textup{Cov\{$Z(\bm{x}),Z(\bm{y})$\}} = \frac{\sigma^2}{2}(\|\bm{x}\|_2^q + \|\bm{y}\|_2^q - \|\bm{x}-\bm{y}\|_2^{q})$. Then $Z(\bm{x})$ is called a \textup{fractional Brownian motion} with \textup{Hurst index} $H = q/2$. We denote this as $Z(\bm{x}) \sim \textup{fBm}(H)$.
\label{def:fbm}
\end{definition}
\noindent With $q=1$ (or $H = 1/2$), $\textup{fBm}(q/2)$ reduces to the standard Brownian motion, and its correlation function reduces to same form assumed in Theorem \ref{thm:sum}. The Hurst index $H = q/2$ plays an important role in controlling the dependency structure of $Z(\bm{x})$: $H>1/2$ indicates positively-correlated increments of $Z$, $H=1/2$ indicates independent increments, and $H<1/2$ indicates negatively-correlated increments. It can be shown \citep{Nua2006} that $|Z(\bm{x}) - Z(\bm{y})| = \mathcal{O}_P(\|\bm{x}-\bm{y}\|_2^{H-\epsilon})$ for any $\epsilon > 0$, so a \textit{smaller} choice of $H$ results in more \textit{volatile} sample paths for $Z(\bm{x})$.

The use of fBm's for computer experiment emulation was first proposed in \cite{ZA2014}, although such models have been employed earlier for intrinsic kriging in the spatial statistics literature. Following \cite{ZA2014}, the non-stationary fBm's enjoy three attractive emulation properties compared to traditional, stationary GPs: they avoid predictions which revert to the mean, ameliorate numerical instabilities from the near-singularity of $\bm{R}_n$, and can better handle abrupt function changes. In other words, these three properties highlight the \textit{robustness} of the fBm model for emulating a broader class of response surfaces. In light of Theorem \ref{thm:sum} and its connection to $\textup{fBm}(1/2)$, the SPs proposed here can then be viewed as optimal information-theoretic designs which \textit{amplify} these robustness properties of fBm emulation. 

\section{Support points as a trade-off between minimax and maximin}
\label{sec:minmax}
Next, we reveal several novel insights which reconcile support points with minimax and maximin designs, two popular designs used in computer experiment emulation. Using the generalized energy distance and the fractional Brownian motion from the previous section, we show that support points can be viewed as a limiting compromise between minimax and maximin designs, with greater weight on the former.

\subsection{The generalized energy distance and fractional Brownian motion}
\label{sec:gen}
Consider first the following generalization of the energy distance:

\begin{definition}[$q$-energy distance; \cite{SZ2013}]
Let $q \in (0,2)$, and adopt the same notation as in Definition \ref{def:energy}, with $\mathbb{E}\|\bm{X}\|^q < \infty$ and $\mathbb{E}\|\bm{Y}\|^q < \infty$. The \textup{$q$-energy distance} between $F$ and $G$ is defined as:
\begin{equation}
\mathcal{E}_q(F,G) \equiv 2\mathbb{E}\|\bm{X}-\bm{Y}\|_2^q - \mathbb{E}\|\bm{Y} - \bm{Y}'\|_2^q - \mathbb{E}\|\bm{X} - \bm{X'}\|_2^q.
\label{eq:gen}
\end{equation}
\label{def:gen}
\end{definition}
\vspace{-0.7cm}
\noindent With $q = 1$, one recovers the original energy distance in Definition \ref{def:energy}. This generalized energy distance was originally proposed in \cite{SZ2013} for testing goodness-of-fit of samples from heavy-tailed distributions.

The following theorem outlines an illuminating connection between the $q$-energy distance and the integration error of functions from the fractional Brownian motion $\textup{fBm}(q/2)$:
\begin{theorem}[Integration error and $q$-energy distance]
Let $\{\bm{x}_i\}_{i=1}^n \subseteq [0,1]^p$ be a point set with corresponding e.d.f. $F_n$. If an integrand $f$ is drawn from $\textup{fBm}(q/2)$ for some $q \in (0,2)$, then:
\begin{equation}
\mathbb{E}[I^2(f;F,F_n)] = \frac{\sigma^2}{2} \mathcal{E}_q(F,F_n).
\label{eq:kh}
\end{equation}
\label{thm:kh}
\end{theorem}
\vspace{-1.0cm}
\begin{proof}
Suppose $f$ is drawn from the fractional Brownian motion $\textup{fBm}(q/2)$ for some $q \in (0,2)$, and consider the following expansion of the expected squared-error $\mathbb{E}[I^2(f;F,F_n)]$:
\small
\begin{align*}
\mathbb{E}[I^2(f;F,F_n)] &= \mathbb{E}\left[ \left( \int_{[0,1]^p} f(\bm{x}) \; d[F-F_n](\bm{x}) \right)^2 \right]\\
&= \mathbb{E}\left[ \int_{[0,1]^p} \int_{[0,1]^p} f(\bm{x}) f(\bm{y}) \; d[F-F_n](\bm{x}) \; d[F-F_n](\bm{y}) \right]\\
&= \int_{[0,1]^p} \int_{[0,1]^p} \mathbb{E}[f(\bm{x}) f(\bm{y})] \; d[F-F_n](\bm{x}) \; d[F-F_n](\bm{y}) \quad \text{(by BCT)}\\
&= \int_{[0,1]^p} \int_{[0,1]^p} \frac{\sigma^2}{2} \left( \|\bm{x}\|_2^q + \|\bm{y}\|_2^q  - \|\bm{x}-\bm{y}\|_2^q\right) \; d[F-F_n](\bm{x}) \; d[F-F_n](\bm{y})\\
& \hspace{7cm} \text{(by Definition \ref{def:fbm} (d))}\\
&= - \frac{\sigma^2}{2}\int_{[0,1]^p} \int_{[0,1]^p} \|\bm{x}-\bm{y}\|_2^q \; d[F-F_n](\bm{x}) \; d[F-F_n](\bm{y}) \\
& \hspace{0.5cm} \text{(since $\int_{[0,1]^p} \|\bm{x}\|_2^q \; d[F-F_n](\bm{y}) = \int_{[0,1]^p} \|\bm{y}\|_2^q \; d[F-F_n](\bm{x}) = 0$)}\\
&= \text{\mbox{\footnotesize $- \frac{\sigma^2}{2} \Bigg( \int_{[0,1]^p} \int_{[0,1]^p} \|\bm{x}-\bm{y}\|_2^q \; dF(\bm{x}) \; dF(\bm{y})  - 2 \int_{[0,1]^p} \int_{[0,1]^p} \|\bm{x}-\bm{y}\|_2^q \; dF(\bm{x}) \; dF_n(\bm{y}) $}} \\
&  \quad \quad \text{\mbox{\footnotesize $+ \int_{[0,1]^p} \int_{[0,1]^p} \|\bm{x}-\bm{y}\|_2^q \; dF_n(\bm{x}) \; dF_n(\bm{y}) \Bigg)$}}\\
&= \frac{\sigma^2}{2} \left( 2\mathbb{E}\|\bm{X}-\bm{Y}\|_2^q - \mathbb{E}\|\bm{Y} - \bm{Y}'\|_2^q - \mathbb{E}\|\bm{X} - \bm{X}'\|_2^q \right) \\
& \hspace{4cm} \text{(where $\bm{X},\bm{X}' \distas{i.i.d.} F_n$ and $\bm{Y},\bm{Y}' \distas{i.i.d.} F$)}\\
&= \frac{\sigma^2}{2} \mathcal{E}_q(F,F_n),
\end{align*}
\normalsize
which is as desired.
\end{proof}

\noindent In other words, when the integrand $f$ follows a fractional Brownian motion with Hurst index $q/2$, the expected squared integration error is proportional to the $q$-energy distance between $F_n$ and $F$. With $q=1$, Proposition \ref{thm:kh} shows that this expected squared error is proportional to the original energy distance in Definition \ref{def:energy}. Support points, by minimizing the latter, is therefore optimal for integrating sample paths from the standard Brownian motion. Proposition \ref{thm:kh} will be useful in the next section for comparing the function spaces of SPs with that for minimax and maximin designs.


\subsection{Minimax and maximin designs}
\label{sec:mm}
Before discussing the connections between SPs, minimax and maximin designs, we briefly review the latter two distance-based designs:
\begin{definition}[Minimax and maximin designs; \cite{Jea1990}]
For a design $\mathcal{D} = \{\bm{x}_i\}_{i=1}^n \subseteq [0,1]^p$, let $Q(\bm{x},\mathcal{D}) = \underset{\bm{z} \in \mathcal{D}}{\textup{Argmin}} \; \|\bm{x} - \bm{z}\|_2$ denote the closest design point to some point $\bm{x} \in [0,1]^p$. A \textup{minimax-distance design} on $[0,1]^p$ (or \textup{minimax design} for short) is defined as:
\begin{equation}
\underset{\mathcal{D} \subseteq [0,1]^p}{\textup{Argmin}} \; \textup{mM}(\mathcal{D}) \equiv \underset{\mathcal{D} \subseteq [0,1]^p}{\textup{Argmin}} \sup_{\bm{x} \in [0,1]^p} \| \bm{x} - Q(\bm{x}, \mathcal{D})\|_2.
\label{eq:minimax}
\end{equation}
A \textup{maximin-distance design} on $[0,1]^p$ (or \textup{maximin design}) is defined as:
\begin{equation}
\underset{\mathcal{D} \subseteq [0,1]^p}{\textup{Argmax}} \; \textup{Mm}(\mathcal{D}) \equiv \underset{\mathcal{D} \subseteq [0,1]^p}{\textup{Argmax}} \min_{i, j} \|\bm{x}_i - \bm{x}_j\|_2.
\label{eq:maximin}
\end{equation}
\end{definition}
\noindent In words, minimax designs \textit{minimize} the \textit{maximum} distance from any point in the design space $[0,1]^p$ to its closest design point, while maximin designs \textit{maximize} the \textit{minimum} distance between any two design points. These designs therefore provide a uniform coverage on $[0,1]^p$ in different ways: a minimax design ensures that the design space is \textit{sufficiently close} to the design, whereas a maximin design ensures design points are \textit{sufficiently spread out} over the design space. Both designs are widely used in many engineering applications, such as computer experiments \citep{Sea2013} and sensor allocation \citep{MJ2017a}.

With this in hand, the following proposition reveals an important connection between the support points for $\mathcal{U}[0,1]^p$ and these two distance-based designs:
\begin{proposition}[SPs as a weighted minimax and maximin design] 
Let $\mathcal{D} = \{\bm{x}_i\}_{i=1}^n \subseteq [0,1]^p$ be a design with e.d.f. $F_n$. For a positive function $h: \Omega \rightarrow \mathbb{R}_{+}$ with $\Omega \subseteq \mathbb{R}^p$ bounded, note that for sufficiently small $q > 0$:
\ben[(a)]
\item If $\Omega$ is measurable with $\textup{Vol}(\Omega) > 0$\footnote{$\textup{Vol}(\Omega)$ denotes the volume of $\Omega$ under Lebesgue measure.}, then:
\begin{equation}
h^q(\bm{z}) \approx \frac{1}{\textup{Vol}(\Omega)} \int_{\Omega} h^q(\bm{z}') \; d\bm{z}' \quad \text{for any $\bm{z} \in \Omega$,}
\label{eq:approxa}
\end{equation}
\item If $\Omega = \{\bm{z}_i\}_{i=1}^m$ is a finite set of points, then:
\begin{equation}
h^q(\bm{z}) \approx \frac{1}{m}\sum_{i=1}^m h^q(\bm{z}_i) \quad \text{for any $\bm{z} \in \Omega$.}
\label{eq:approxb}
\end{equation}
\een
Under such approximations, the $q$-energy between $F$ and $F_n$ becomes:
\begin{equation}
\mathcal{E}_q(F,F_n) \approx 2[\textup{mM}(\mathcal{D})]^q - [\textup{Mm}(\mathcal{D})]^q - \mathbb{E} \|\bm{Y} - \bm{Y}'\|_2^q, \quad \bm{Y}, \bm{Y}' \distas{i.i.d.} F,
\label{eq:wtmm}
\end{equation}
for sufficiently small $q>0$.
\label{thm:dist}
\end{proposition}

\begin{proof}
Rewrite the $q$-energy distance as follows:
\begin{equation}
\mathcal{E}_q(F,F_n) = \frac{2}{n} \sum_{i=1}^n \mathbb{E}\|\bm{x}_i - \bm{Y}\|_2^q - \frac{1}{n^2} \sum_{i=1}^n \sum_{j=1}^n \|\bm{x}_i - \bm{x}_j\|_2^q - \mathbb{E}\|\bm{Y} - \bm{Y}'\|_2^q.
\label{eq:terms}
\end{equation}
Consider the second term on the right-hand side of \eqref{eq:terms}. Using the approximation in \eqref{eq:approxb}, we have for sufficiently small $q > 0$ that:
\begin{equation}
\frac{1}{n^2} \sum_{i=1}^n \sum_{j=1}^n \|\bm{x}_i - \bm{x}_j\|_2^q \approx \min_{i,j} \|\bm{x}_i - \bm{x}_j\|_2^q = [\textup{Mm}(\mathcal{D})]^q.
\label{eq:left}
\end{equation}

Consider next the first term on the right-hand side of \eqref{eq:terms}. Let $\text{Vor}(\bm{x}_i) = \{\bm{x} \in [0,1]^p\; : Q(\bm{x},\mathcal{D}) = \bm{x}_i\}$ be the Voronoi region for design point $\bm{x}_i$, i.e., the set of points in $[0,1]^p$ closer to $\bm{x}_i$ than any other design point. Further, for a given design point $\bm{x}_i$, let $\bm{z}_i = \mathop{\mathrm{Argmax}}_{\bm{z} \in \text{Vor}(\bm{x}_i)} \|\bm{z} - \bm{x}_i\|_2$ be the farthest point (in Euclidean distance) to $\bm{x}_i$ in its Voronoi region. By definition, $\bm{z}_i$ must be on $\partial \text{Vor}(\bm{x}_i)$, the boundary of $\text{Vor}(\bm{x}_i)$. Letting $\overline{\text{Vor}(\bm{x}_i)} \equiv [0,1]^p \setminus {\text{Vor}(\bm{x}_i)}$, we have for sufficiently small $q > 0$ that:
\begin{align*}
\frac{2}{n} \sum_{i=1}^n \mathbb{E}\|\bm{x}_i - \bm{Y}\|_2^q &= \frac{2}{n} \sum_{i=1}^n \int_{[0,1]^p} \|\bm{x}_i - \bm{y}\|_2^q \; d \bm{y} \notag\\
& = \frac{2}{n} \sum_{i=1}^n \left( \int_{{\text{Vor}(\bm{x}_i)}} \|\bm{x}_i - \bm{y}\|_2^q \; d \bm{y} + \int_{\overline{\text{Vor}(\bm{x}_i)} \cup \partial \text{Vor}(\bm{x}_i)} \|\bm{x}_i - \bm{y}\|_2^q \; d \bm{y} \right) \\
& \hspace{5cm} \text{(since $\textup{Vol}(\partial \text{Vor}(\bm{x}_i)) = 0$)}\\
& \approx \frac{2}{n} \sum_{i=1}^n \Big( \textup{Vol}\{\text{Vor}(\bm{x}_i)\} \|\bm{x}_i - \bm{z}_i\|_2^q + \textup{Vol}\{\overline{\text{Vor}(\bm{x}_i)} \cup \partial \text{Vor}(\bm{x}_i)\} \|\bm{x}_i - \bm{z}_i\|_2^q\Big)\\
& \hspace{7.75cm} \text{(by \eqref{eq:approxa})}\\
& = \frac{2}{n} \sum_{i=1}^n \|\bm{z}_i - \bm{x}_i\|_2^q\\
& \hspace{1.5cm} \text{($\textup{Vol}\{\text{Vor}(\bm{x}_i)\} + \textup{Vol}\{\overline{\text{Vor}(\bm{x}_i)} \cup \partial \text{Vor}(\bm{x}_i)\} = 1$)}\\
& \approx 2 \max_{i} \|\bm{z}_i - \bm{x}_i\|_2^q \hspace{4.5cm} \text{(by \eqref{eq:approxb})}\\
& = 2 \sup_{\bm{x} \in [0,1]^p} \| \bm{x} - Q(\bm{x}, \mathcal{D})\|_2^q \hspace{0.5cm} \text{(by definition of $\bm{z}_i$)}\\
& = 2 [\textup{mM}(\mathcal{D})]^q. \hspace{2.4cm} \text{(by the definition in \eqref{eq:minimax})}
\end{align*}
Plugging the above result and \eqref{eq:left} into \eqref{eq:terms}, we obtain the desired result.
\end{proof}

\noindent Proposition \ref{thm:dist} shows that for sufficiently small choices of $q$, an approximation of the $q$-energy distance $\mathcal{E}_q(F,F_n)$ is two-parts the minimax criterion $\textup{mM}(\mathcal{D})$, minus one-part the maximin criterion $\textup{Mm}(\mathcal{D})$, neglecting constants. Employing this approximation as a surrogate objective for minimizing $\mathcal{E}_q(F,F_n)$, the resulting optimal design should therefore jointly \textit{minimize} $\textup{mM}(\mathcal{D})$ and \textit{maximize} $\textup{Mm}(\mathcal{D})$ as $q \rightarrow 0^+$. Setting $q=1$ as a approximation for the limit $q \rightarrow 0^+$, \textit{the support points which minimize $\mathcal{E}(F,F_n)$ can be viewed as designs which are two-parts minimax and one-part maximin.} This trade-off yields several valuable insights on SPs as an experimental design, which we describe below.

The first insight is the implication that this trade-off has on the \textit{approximation error} of the response surface $f(\cdot)$. Recently, there has been interesting theoretical work (see, e.g., \cite{HQ2011}) on analyzing and controlling different sources of error incurred in emulation. When a Gaussian process is used as the underlying stochastic model for the unknown function $f$, two key errors arise: \textit{nominal} errors and \textit{numeric} errors \citep{HQ2011}. Here, nominal errors refer to the inaccuracies between the true function $f$ and its idealized approximation from the GP (not subject to floating point errors), and numeric errors refer to the inaccuracies arising from floating point errors in numerical computations. This body of work suggests that nominal errors can be controlled by reducing $\textup{mM}(\mathcal{D})$ (also known as the \textit{fill distance} of the design), while numeric errors can be controlled by increasing $\textup{Mm}(\mathcal{D})$ (or \textit{separation distance} of the design).

In light of this, the minimax and maximin designs introduced earlier reduce each of these two error sources \textit{separately}, by minimizing fill distance and maximizing separation distance, respectively. However, as noted in \cite{MJ2017a}, the minimization of $\textup{mM}(\mathcal{D})$ and the maximization of $\textup{Mm}(\mathcal{D})$ can be at odds; a minimax design need not be maximin, nor vice versa. To jointly control \textit{both} nominal and numeric errors, a criterion which compromises between minimax and maximin is necessary. To this end, Proposition \ref{thm:dist} shows that \textit{SPs provide this desired error trade-off, with greater emphasis on controlling nominal errors}. The greater focus on nominal over numeric error is quite intuitive, because from interpolation theory (see Section 11.2 of \cite{Wen2004}), it is known that the fill distance plays a dominant role in controlling overall error.

Further insight can be gained by exploring the underlying function space for the response surface $f$. From Section \ref{sec:fbm}, we know that a fractional Brownian motion with smaller Hurst index $H=q/2$ results in more \textit{volatile} sample paths, due to an increasing negative correlation between process increments. Propositions \ref{thm:kh} and \ref{thm:dist} then show that a weighted combination of a minimax and maximin design is optimal (in a limiting sense) for integrating such volatile sample paths. This appeal to more ``wiggly'' functions is quite similar to the initial motivation for minimax and maximin designs in \cite{Jea1990}. There, the authors showed that minimax and maximin designs are $G$-optimal and $D$-optimal, respectively, when $f$ is drawn from a Gaussian process with correlation parameter $\theta \rightarrow 0^+$ (i.e., a near-independent Gaussian process). Viewed this way, \textit{SPs can be seen as optimal designs which trade-off between $G$-optimality and $D$-optimality} -- it jointly minimize the worst-case prediction error and maximize the Shannon information gain on the response surface $f$. Again, because the primary objective in emulation is to improve predictive accuracy, the greater emphasis on $G$-optimality is quite appealing.


\section{Further connection to existing designs}
\label{sec:con}
Next, we discuss further connections between SPs, PSPs and two other experimental designs: the uniform design \citep{Fan1980}, and the maximum projection (MaxPro) design \citep{Jea2015}. Specifically, we demonstrate several computational and modeling advantages of the proposed designs over these existing methods.

\subsection{Uniform designs}
\label{sec:unifdes}
Consider first the uniform designs (UDs) proposed by \cite{Fan1980}, which minimize some measure of discrepancy between the uniform distribution $F = \mathcal{U}[0,1]^p$ and the empirical distribution of the design, $F_n$. The uniform designs in \textsf{JMP} 13 \citep{SAS2016} minimize the so-called centered-$\mathcal{L}_2$ discrepancy \citep{Hic1998}:
\begin{align}
\begin{split}
CL_2(F,F_n) &= \sum_{0 \neq \bm{u} \subseteq [p]} \int_{[0,1]^{\bm{u}}} \left| F(\bm{x}) - F_n(\bm{x}) \right|^2 \; d\bm{x},
\end{split}
\label{eq:cl2}
\end{align}
which measures the $\mathcal{L}_2$-discrepancy between $F$ and $F_n$ in all possible subspaces of $[0,1]^p$. By minimizing such a criterion, UDs ensure the empirical distribution function $F_n$ well-approximates the desired uniform distribution $F$ in both the full $p$-dimensional cube, as well as its projected subspaces.

The connection between UDs and SPs can be seen by rewriting \eqref{eq:cl2} as the following kernel discrepancy (see, e.g., \cite{Hic1998}):
\begin{small}
\begin{equation}
\int_{[0,1]^p} \int_{[0,1]^p} \prod_{l=1}^p \left( 1 + \frac{|x_l-0.5|}{2} + \frac{|y_l-0.5|}{2} - \frac{|x_l-y_l|}{2} \right) d[F-F_n](\bm{x}) \; d[F-F_n](\bm{y}),
\label{eq:cl2kern}
\end{equation}
\end{small}
and rewriting the energy distance $\mathcal{E}(F,F_n)$ in \eqref{eq:energy} as:
\begin{equation}
\int_{[0,1]^p} \int_{[0,1]^p} -\|\bm{x} - \bm{y}\|_2 \; d[F-F_n](\bm{x}) \; d[F-F_n](\bm{y}).
\label{eq:enkern}
\end{equation}
Viewed this way, both UDs and SPs can be interpreted as designs which minimize some kernel-based discrepancy \citep{Hic1998}, with the kernel for UDs being the centered product term within the integral in \eqref{eq:cl2kern}, and the kernel for SPs being the negative Euclidean norm. Using some form of the Koksma-Hlawka inequality (see Theorem \ref{thm:interr}a and \cite{Dea2013}), both SPs and UDs are optimal for integration under an appropriately set function space.

Despite this similarity, SPs have two important advantages over UDs as computer experiment designs. First, in employing a \textit{distance-based kernel}, SPs have an inherent connection to existing \textit{distance-based designs}. As shown in previous sections, this then allows us to justify SPs as (a) robust designs for GP modeling, and (b) designs which jointly minimize two sources of approximation error. UDs, on the other hand, do not make use of such a distance-based kernel, and there is therefore little justification for why such designs are appropriate for functional approximation (this lack of justification is further echoed in \cite{Wie1991} and \cite{Sea2013}).

Second, from a practical perspective, SPs can be generated much more efficiently than UDs, the latter being computationally expensive to generate even for small designs in moderate dimensions. For example, a $(n,p)=(100,10)$ UD requires several \textit{hours} of computation using \textsf{JMP} 13! To contrast, SPs can be computed in a fraction of this time by exploiting the difference-of-convex structure in $\mathcal{E}(F,F_n)$ for optimization. For the same $(n,p) = (100,10)$ example, SPs can be generated within several \textit{seconds} \citep{MJ2017b}, and a much larger, higher-dimensional $(n,p) = (10000,500)$ design can be generated within 3 hours. Section \ref{sec:alg} outlines more computational comparisons and improved algorithms for generating SPs and PSPs.

\subsection{Maximum projection designs}
\label{sec:maxpro}
Consider next the MaxPro designs \citep{Jea2015}, which minimizes the following criterion with $\lambda = 0$:
\begin{equation}
\text{MaxPro}(\mathcal{D}; \lambda) = \sum_{i=1}^{n-1} \sum_{j=i+1}^n \frac{1}{\prod_{l=1}^p \left\{ (x_{i,l} - x_{j,l})^2 + \lambda \right\} }, \quad \bm{x}_i = (x_{i,l})_{l=1}^p.
\label{eq:maxpro}
\end{equation}
This so-called MaxPro criterion can be derived as the maximin criterion in \eqref{eq:maximin}, subject to a prior distribution which dictates effect sparsity -- the belief that only several variables are active in a high-dimensional function. As such, MaxPro designs enjoy excellent maximin performance on both the full space $[0,1]^p$, as well as on its projected subspaces.

An insightful connection between PSPs and MaxPro designs is revealed by assuming a simpler anisotropic Gaussian kernel for \eqref{eq:kern}, with scale parameters $\theta_l$ following i.i.d. $Exp(\lambda)$ priors, where $\lambda$ is a rate parameter. The motivation for this is again effect sparsity; under such a prior, only several $\theta_l$'s are large, which suggests the underlying response surface is dominated by only several key effects. Subject to such a prior, the following theorem gives a closed-form expression of the PSP objective function \eqref{eq:psp}:
\begin{theorem}[PSPs and MaxPro designs]
Consider the generalized Gaussian kernel in \eqref{eq:kern}, with $\theta_{\bm{u}} = \theta_l$ for $\bm{u} = \{l\}$, and 0 otherwise. If $\theta_l \distas{i.i.d.} Exp(\lambda)$, then the objective in \eqref{eq:psp} is proportional to:
\begin{equation}
\textup{MaxPro}(\mathcal{D};\lambda) - \frac{n}{\lambda^{p/2}} \sum_{i=1}^n \prod_{l=1}^p \phi(x_{i,l};\lambda) + \frac{n}{2\lambda^p},
\label{eq:maxproadj}
\end{equation}
where $\phi(x;\lambda) = \textup{tan}^{-1}\left({x}/{\sqrt{\lambda}}\right) + \textup{tan}^{-1}\left((1-x)/{\sqrt{\lambda}}\right)$.
\label{thm:maxpro}
\end{theorem}

\begin{proof}
Take the Gaussian kernel in \eqref{eq:kern}, with $\theta_{\bm{u}} = \theta_l$ for $\bm{u}=\{l\}$ and 0 otherwise. With $\bm{Y} \sim F = \mathcal{U}[0,1]^p$ and $\theta_l \distas{i.i.d.} Exp(\lambda)$, one can then show:
\begin{align}
\mathbb{E}_{\boldsymbol{\theta}\sim\pi}[\gamma_{\boldsymbol{\theta}}(\bm{x}_i,\bm{x}_j)] &= \int_{\mathbb{R}_+^p} \exp\left\{ - \sum_{l=1}^p \theta_l (x_{i,l} - x_{j,l})^2 \right\} \prod_{l=1}^p \Big\{ \lambda \exp(-\lambda\theta_l) \Big\} \; d\boldsymbol{\theta} \notag\\
&= \lambda^p \prod_{l=1}^p \left( \int_{\mathbb{R}_+} \exp\left\{ - \theta_l [(x_{i,l} - x_{j,l})^2 + \lambda] \right\} \; d\theta_l \right) \notag\\
& = \frac{\lambda^p}{\prod_{l=1}^p \left( [x_{i,l} - x_{j,l}]^2 + \lambda \right)}, \notag
\end{align}
and:
\begin{align*}
\mathbb{E}_{\bm{Y},\boldsymbol{\theta}\sim\pi}[\gamma_{\boldsymbol{\theta}}(\bm{x}_i,\bm{Y})] &= \mathbb{E}_{\bm{Y}} \mathbb{E}_{\boldsymbol{\theta} \sim \pi} [\gamma_{\boldsymbol{\theta}}(\bm{x}_i,\bm{Y})] \notag \\
&= \lambda^p \prod_{l=1}^p \left( \mathbb{E}_{Y_l}\left[ \frac{1}{ (x_{i,l} - Y_{l})^2 + \lambda } \right] \right) \quad \text{($Y_l$'s are i.i.d.)}\\
&= \lambda^{p/2}\prod_{l=1}^p \phi(x_{i,l};\lambda),
\end{align*}
where $\phi(x;\lambda) = \textup{tan}^{-1}\left({x}/{\sqrt{\lambda}}\right) + \textup{tan}^{-1}\left((1-x)/{\sqrt{\lambda}}\right)$. Using these two results, the objective function in \eqref{eq:psp} becomes:
\[\frac{2\lambda^p}{n^2} \left\{ - \frac{n}{\lambda^{p/2}} \sum_{i=1}^n \prod_{l=1}^p \phi(x_{i,l};\lambda) + \sum_{i=1}^{n-1} \sum_{j=i+1}^n \frac{1}{\prod_{l=1}^p \left( [x_{i,l} - x_{j,l}]^2 + \lambda \right)} + \frac{n}{2\lambda^p} \right\} ,\]
which completes the proof.
\end{proof}

\begin{figure}[t]
\centering
\includegraphics[width=0.7\textwidth]{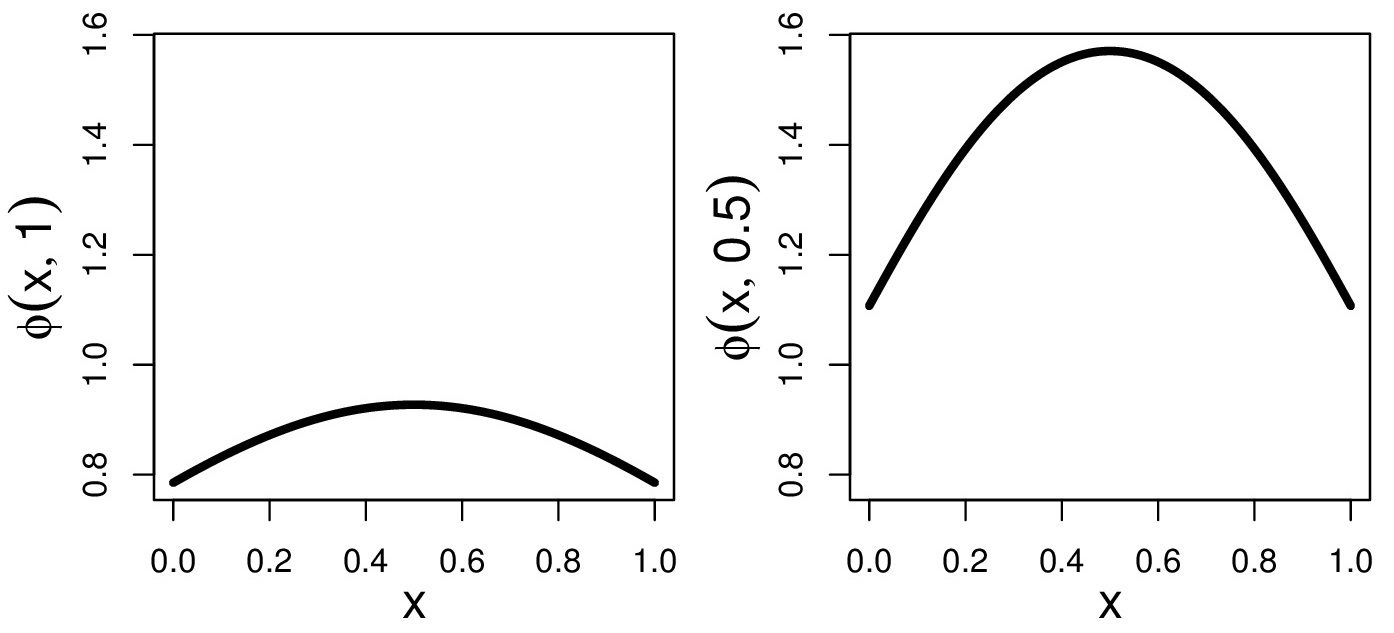}
\caption{A visualization of the correction terms $\phi(x;1)$ and $\phi(x;0.5)$. The maximization of this term pushes design points away from the boundaries of $[0,1]$.}
\label{fig:phi}
\end{figure}

In other words, subject to the sparsity prior $\theta_l \distas{i.i.d.} Exp(\lambda)$, the resulting PSPs jointly \textit{minimize} the MaxPro criterion in \eqref{eq:maxpro}, and \textit{maximize} an expression involving the correction term $\phi(x;\lambda)$. For illustration, Figure \ref{fig:phi} plots $\phi(x;\lambda)$ for $\lambda = 1$ and $\lambda = 0.5$. Note that $\phi(x;\lambda)$ has two minima on the boundary points of $x=0$ and $x=1$, and is maximized at the midpoint $x=0.5$. The maximization of the expression involving $\phi(x;\lambda)$ in \eqref{eq:maxproadj} therefore pushes design points away from the boundaries of $[0,1]^p$. This feature is particularly appealing in light of the fact that MaxPro designs are known to push design points towards boundaries \citep{Jea2015}. Theorem \ref{thm:maxpro} shows that PSPs \textit{correct} this boundary attraction behavior of MaxPro designs, thereby allowing for \textit{better uniformity} of the design on projected subspaces.

Following \cite{MJ2017c}, the PSPs employed in this paper assume the so-called product-and-order (POD) form for the generalized kernel in \eqref{eq:kern}:
\begin{equation}
\theta_{\bm{u}} = \Gamma_{|\bm{u}|}\prod_{l \in \bm{u}} \theta_l, \quad \Gamma_{|\bm{u}|} = p^{-1/4} (|\bm{u}|!)^{-1/2}, \quad \theta_l \distas{i.i.d.} Gamma(0.1,1).
\label{eq:pod}
\end{equation}
Here, the so-called \textit{product weights} $(\theta_l)_{l=1}^p$ quantify the importance of each of the $p$ variables, while the so-called \textit{order weights} $(\Gamma_{|\bm{u}|})_{|\bm{u}|=1}^\infty$ quantify the significance of effects with order $|\bm{u}|$. The POD form in \eqref{eq:pod} provides an appealing framework for encoding the three principles of experimental design (see Section \ref{sec:psp}). In particular, \textit{effect hierarchy} -- the belief that lower-order effects are more likely active than higher-order ones -- is imposed here by setting a factorial decay on the \textit{order} weights $(\Gamma_{|\bm{u}|})_{|\bm{u}|=1}^\infty$. Likewise, (strong) \textit{effect heredity} -- the belief that a higher-order effect is active only when all its lower-order component effects are active -- is imposed by the \textit{product} structure $\prod_{l \in \bm{u}} \theta_l$ in the POD specification. As before, \textit{effect sparsity} is imposed via an i.i.d. prior distribution on the product weights $(\theta_l)_{l=1}^p$. In this view, the PSPs generated under the POD specification \eqref{eq:pod} yield optimal design schemes for learning the response surface under effect hierarchy, effect heredity and effect sparsity. We refer the reader to \cite{MJ2017c} for a more detailed discussion.

\section{Optimization algorithms}
\label{sec:alg}
Before presenting simulations, we briefly review the algorithms used for generating SPs and PSPs, following \cite{MJ2017b,MJ2017c}, respectively. These algorithms exploit several appealing properties in the formulations for SPs and PSPs for efficient optimization. We provide below a high-level introduction of these methods, then present two key improvements of such methods for design optimization.

\begin{figure}[!t]
\begin{algorithm}[H]
\caption{\texttt{sp.sccp}: Computing support points}
\label{alg:spsccp}
\begin{algorithmic}
\small
\stb Warm-start $\mathcal{D}^{[0]} = \{\bm{x}_i^{[0]}\}_{i=1}^n$ using a design-of-choice.
\stb Set $k = 0$, and \textbf{repeat} until convergence:
\bi
\item Resample $\mathcal{Y}^{[k]} \distas{i.i.d.} \mathcal{U}[0,1]^p$ (or from a randomized Sobol' sequence).
\item \textbf{For} $i = 1, \cdots, n$ \textbf{do parallel}:
\bi
\item $\bm{x}_i^{[k+1]} \leftarrow M_i(\mathcal{D}^{[k]}; \mathcal{Y}^{[k]})$.
\ei
\item Update $\mathcal{D}^{[k+1]} \leftarrow \{\bm{x}_i^{[k+1]}\}_{i=1}^n$, set $k \leftarrow k + 1$.
\ei
\stb Return the converged design $\mathcal{D}^{[\infty]}$.
\normalsize
\end{algorithmic}
\end{algorithm}
\vspace{-0.65cm}
\begin{algorithm}[H]
\caption{\texttt{psp.sccp}: Computing projected support points}
\label{alg:pspsccp}
\begin{algorithmic}
\small
\setlength{\leftmargini}{10pt}
\stb Warm-start $\mathcal{D}^{[0]} = \{\bm{x}_i^{[0]}\}_{i=1}^n$ using a design-of-choice.
\stb Set $k = 0$, and \textbf{repeat} until convergence:
\bi
\item \textbf{For} $i = 1, \cdots, n$:
\bi
\item Resample $\mathcal{Y}^{[k]} \distas{i.i.d.} \mathcal{U}[0,1]^p$ (or from a randomized Sobol' sequence) and $\vartheta^{[k]} \distas{i.i.d.} \pi$.
\item $\bm{x}_i^{[k+1]} \leftarrow \widetilde{M}_i(\bm{x}_i^{[k]}; \mathcal{Y}^{[k]}, \vartheta^{[k]}, \mathcal{D}^{[k]}_{-i})$.
\item Update $\mathcal{D}^{[k]}_{i} \leftarrow \bm{x}_i^{[k+1]}$.
\ei
\item Update $\mathcal{D}^{[k+1]} \leftarrow \{\bm{x}_i^{[k+1]}\}_{i=1}^n$, set $k \leftarrow k + 1$.
\ei
\stb Return the converged design $\mathcal{D}^{[\infty]}$.
\normalsize
\end{algorithmic}
\end{algorithm}
\vspace{-0.65cm}
\end{figure}

Consider first the optimization problem for SPs in \eqref{eq:supp}, which can be viewed as a \textit{difference-of-convex} (d.c.) program (i.e., where the objective to minimize can be written as a difference of convex functions). To solve this d.c. program, \cite{MJ2017b} proposed an efficient optimization algorithm called \texttt{sp.sccp}, which makes use of the \textit{convex-concave procedure} (CCP, \cite{YR2003}). The key idea behind CCP is first to replace the concave term in the d.c. program with a convex upper bound, then solve the resulting ``surrogate'' formulation (which is now convex) using convex programming methods. With this in hand, \texttt{sp.sccp} iteratively applies the following two steps. First, the expectation in \eqref{eq:supp} is approximated using a large Monte Carlo batch sample $\mathcal{Y} = \{\bm{y}_m\}_{m=1}^N \sim F$. Next, using a quadratic upper bound on the concave term $-\|\bm{x}_i - \bm{x}_j\|_2$, this approximate formulation can be solved as a closed-form expression $M_i(\cdot; \mathcal{Y})$ using CCP (see \cite{MJ2017b} for precise form of $M_i$). By iterating these two steps, one can then show that the sequence of generated designs $(\mathcal{D}^{[k]})_{l=0}^\infty$ converges to a local optimum of \eqref{eq:supp} in $\mathcal{O}(n^2p)$ work. Algorithm \ref{alg:spsccp} summarizes these iterative steps for \texttt{sp.sccp}.

Similarly, an algorithm called \texttt{psp.sccp} is employed in \cite{MJ2017c} for generating PSPs. \texttt{psp.sccp} makes use of the fact that, for a given design point $\bm{x}_i$ (with other points fixed), the optimization in \eqref{eq:psp} for $\bm{x}_i$ can be viewed as a near-d.c. program. As before, \texttt{psp.sccp} iterates the following two steps. First, for each $\bm{x}_i$, the expectations in \eqref{eq:psp} are approximated using two Monte Carlo batch samples: the first batch $\mathcal{Y} = \{\bm{y}_m\}_{m=1}^N$ from distribution $F$, and the second batch $\vartheta = \{\boldsymbol{\theta}_r\}_{r=1}^R$ from prior $\pi$. Next, letting $\mathcal{D}_{-i}$ denote the fixed remaining $n-1$ points, the approximate formulation can then be solved as a closed-form expression $\widetilde{M}_i(\cdot; \mathcal{Y},\vartheta,\mathcal{D}_{-i})$ using CCP (see \cite{MJ2017c} for precise form of $\widetilde{M}_i$). Iterating this cyclically over all design points, the sequence of generated designs $(\mathcal{D}^{[k]})_{l=0}^\infty$ again converges to a local optimum of \eqref{eq:psp}. Algorithm \ref{alg:pspsccp} summarizes these iterative steps for \texttt{psp.sccp}.

We propose here two key improvements to these algorithms for design optimization. The first is improving on the Monte Carlo \textit{initialization} of the starting design $\mathcal{D}^{[0]}$, as originally suggested in \cite{MJ2017c,MJ2017b}. Recall that for local optimization, the choice of an \textit{initial} solution can have a large impact on the quality of the final solution. Here, this means the choice of initial design $\mathcal{D}^{[0]}$ in \texttt{sp.sccp} or \texttt{psp.sccp} can greatly influence the quality of the final design $\mathcal{D}^{[\infty]}$. Given the wealth of existing designs in the literature, we make use of such designs to provide an \textit{informed} initialization of \texttt{sp.sccp} or \texttt{psp.sccp}. Put another way, \texttt{sp.sccp} or \texttt{psp.sccp} can be viewed as a \textit{post-processing} step which improves space-filling properties of an existing design, all the while retaining desired characteristics of this original design. As indicated in Algorithms \ref{alg:spsccp} and \ref{alg:pspsccp}, we recommend both methods be initialized with a user's design-of-choice, whether that be a favorite space-filling design or a low-discrepancy point set. Compared to the original Monte Carlo initialization, such a modification yields better quality final designs in our experience, and offers users the flexibility of customizing SPs and PSPs to their design-of-choice. The second improvement is on the batch sampling step $\mathcal{Y} \sim F$. For the specific case of $F = \mathcal{U}[0,1]^p$ here, we found that batch sampling from randomized low-discrepancy sequences (specifically, the randomized Sobol' sequences in \cite{Owe1998}) can yield improved optimization performance to the original Monte Carlo batch sampling in \cite{MJ2017c, MJ2017b}. This is not too surprising, because the former enjoys improved sampling properties to the latter. This modified batch sampling step is bracketed in Algorithms \ref{alg:spsccp} and \ref{alg:pspsccp}.

\begin{figure}[t]
\centering
\includegraphics[width=0.8\textwidth]{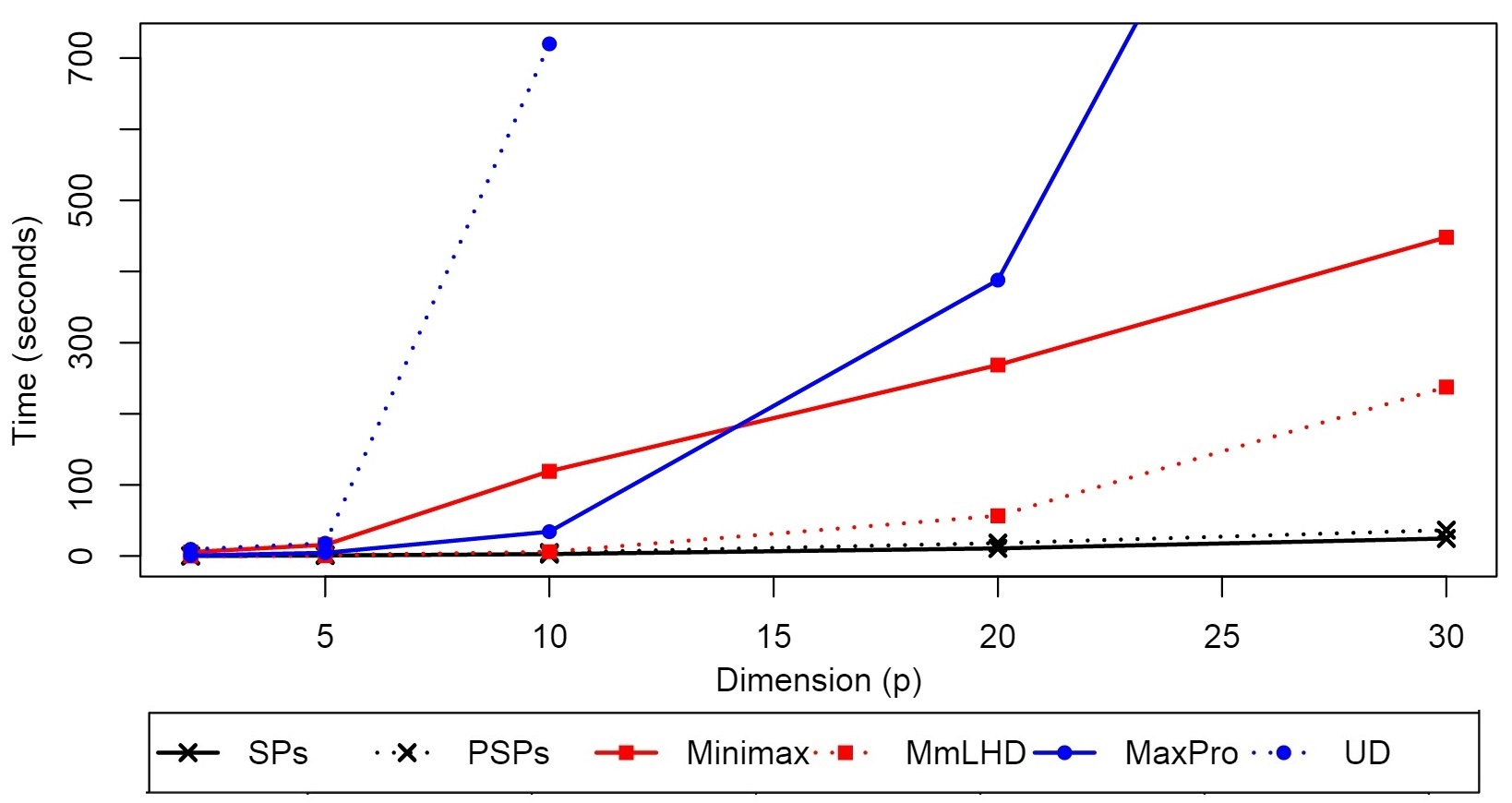}
\caption{Computation time (in seconds) for $n=10p$-point designs on $[0,1]^p$, for various dimensions $p$.}
\label{fig:comptime}
\end{figure}

To illustrate the computational efficiency of \texttt{sp.sccp} and \texttt{psp.sccp} over existing designs, Figure \ref{fig:comptime} shows the computation times (in seconds) required for generating $n=10p$-point designs in $[0,1]^p$ (following the run size recommendation in \cite{Loe2009}), with $p$ as large as 30. All computations were performed on a 2.60Ghz quad-core processor; see Section \ref{sec:vis} for further implementation details. We see that SPs, PSPs and MmLHDs are the quickest designs to generate for small run sizes in low dimensions, with SPs and PSPs much more efficient than existing design algorithms for large run sizes in high dimensions. By exploiting the appealing optimization structure in \eqref{eq:supp} and \eqref{eq:psp}, SPs and PSPs can therefore be generated more efficiently than existing designs (particularly for large run sizes in high dimensions), all the while enjoying improved emulation performance as well (see Sections \ref{sec:eff} and \ref{sec:app}).


\section{Simulations}
\label{sec:sim}

We now present some simulations which explore the effectiveness of SPs and PSPs as computer experiment designs. We first visually inspect the proposed designs, then investigate their distanced-based metrics and robustness for Gaussian process modeling.

\subsection{Visualization}
\label{sec:vis}
\begin{figure}[t]
\centering
\includegraphics[width=\linewidth]{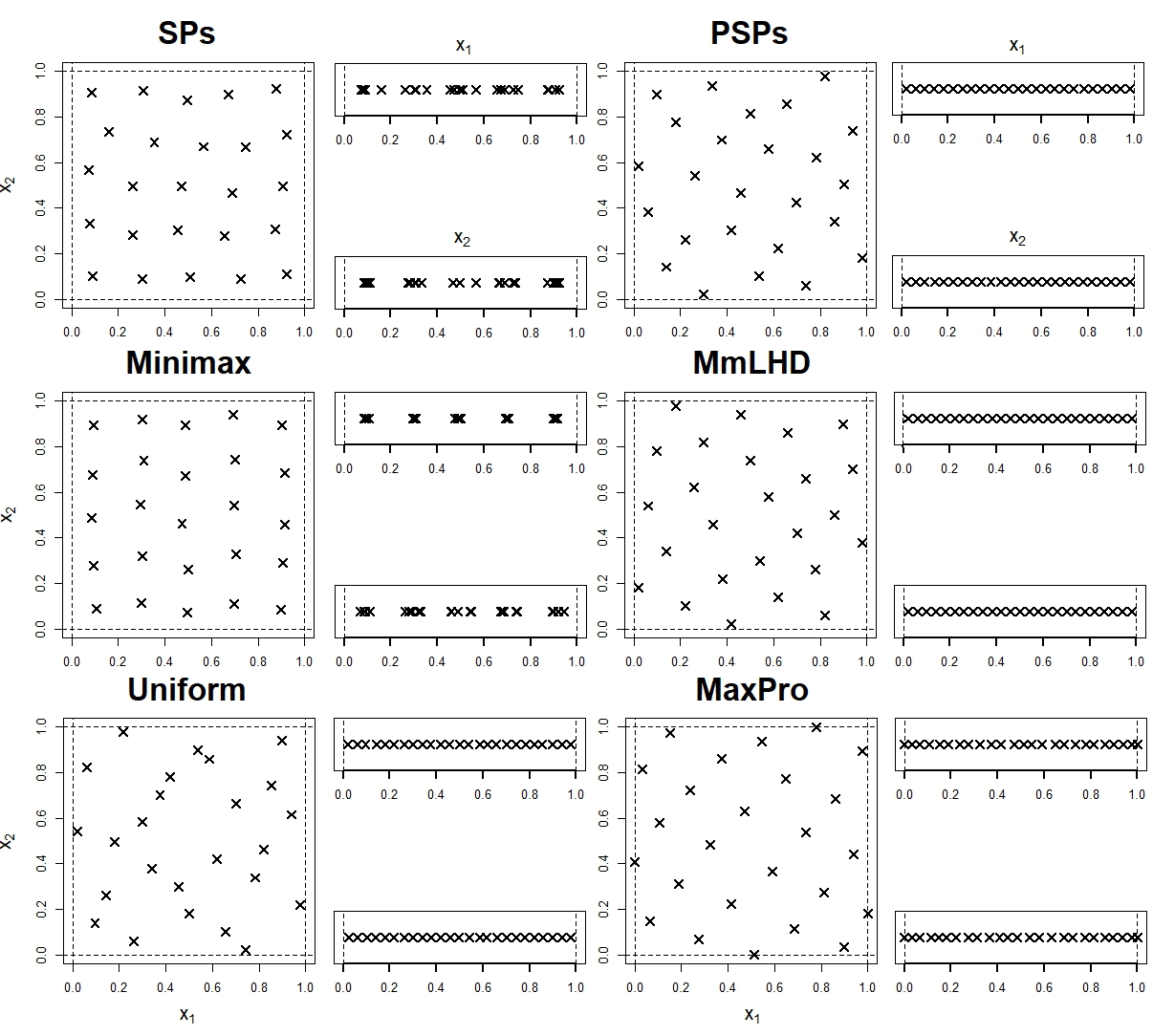}
\caption{$n=25$-point designs for SPs and PSPs on $F = \mathcal{U}[0,1]^2$, and the minimax, MmLHD, uniform and MaxPro designs. (Left) 2-d visualization; (Right) 1-d projections.}
\label{fig:visunif}
\end{figure}

We first visualize the SPs and PSPs on $F = \mathcal{U}[0,1]^p$, and four popular computer experiment designs in the literature: minimax designs (generated using the \textsf{R} package \texttt{minimaxdesign} \cite{Mak2016}), maximin-distance Latin hypercube designs (MmLHDs \cite{MM1995}; using the \textsf{R} package \texttt{SLHD} \cite{Ba2015}), uniform designs (using \textsf{JMP} 13 \cite{SAS2016}) and MaxPro designs (using the \textsf{R} package \texttt{MaxPro} \cite{BJ2015}). The SPs and PSPs used here are generated using \texttt{sp.sccp} and \texttt{psp.sccp}, and are initialized using randomized Sobol' sequences \citep{Owe1998} and MaxPro designs, respectively.

Figure \ref{fig:visunif} plots the $n=25$-point designs on the 2-d hypercube $[0,1]^2$, as well as its 1-d projections onto the two coordinate axes. Four interesting observations can be made here. First, the minimax design appears to be (nearly) a full factorial design with five levels on each factor, and enjoys good space-fillingness on the full space $[0,1]^2$ but poor coverage after projections. Visually, SPs look quite similar to the minimax design but with points slightly pushed towards boundaries, which is in line with the observation that SPs are a weighted combination of minimax and maximin designs (Section \ref{sec:mm}), with greater emphasis on the former. Second, comparing SPs to PSPs, PSPs enjoy greatly improved projections in 1-d by sacrificing some space-fillingness in 2-d. Hence, SPs should be used when \textit{all} factors are known to be active in the computer experiment, whereas PSPs should be used when \textit{some} factors are known to be less active. Third, comparing SPs and PSPs with the UD, we see that the former has much better space-fillingness in 2-d. This demonstrates the appeal of a \textit{distance-based} kernel (see Section \ref{sec:unifdes}); the use of such a kernel in the energy distance $\mathcal{E}$ yields good space-filling properties, while the lack of such a kernel in $CL_2$ results in poor space-fillingness. Lastly, the projection of the MaxPro design onto 1-d appears to push points slightly towards the boundaries of $[0,1]$; PSPs appear to correct this boundary attraction behavior of MaxPro, which is in line with the correction term $\phi(x;\lambda)$ in Theorem \ref{thm:maxpro}.

\subsection{Distance-based design metrics}
\begin{figure}[t]
\centering
\includegraphics[width=0.70\textwidth]{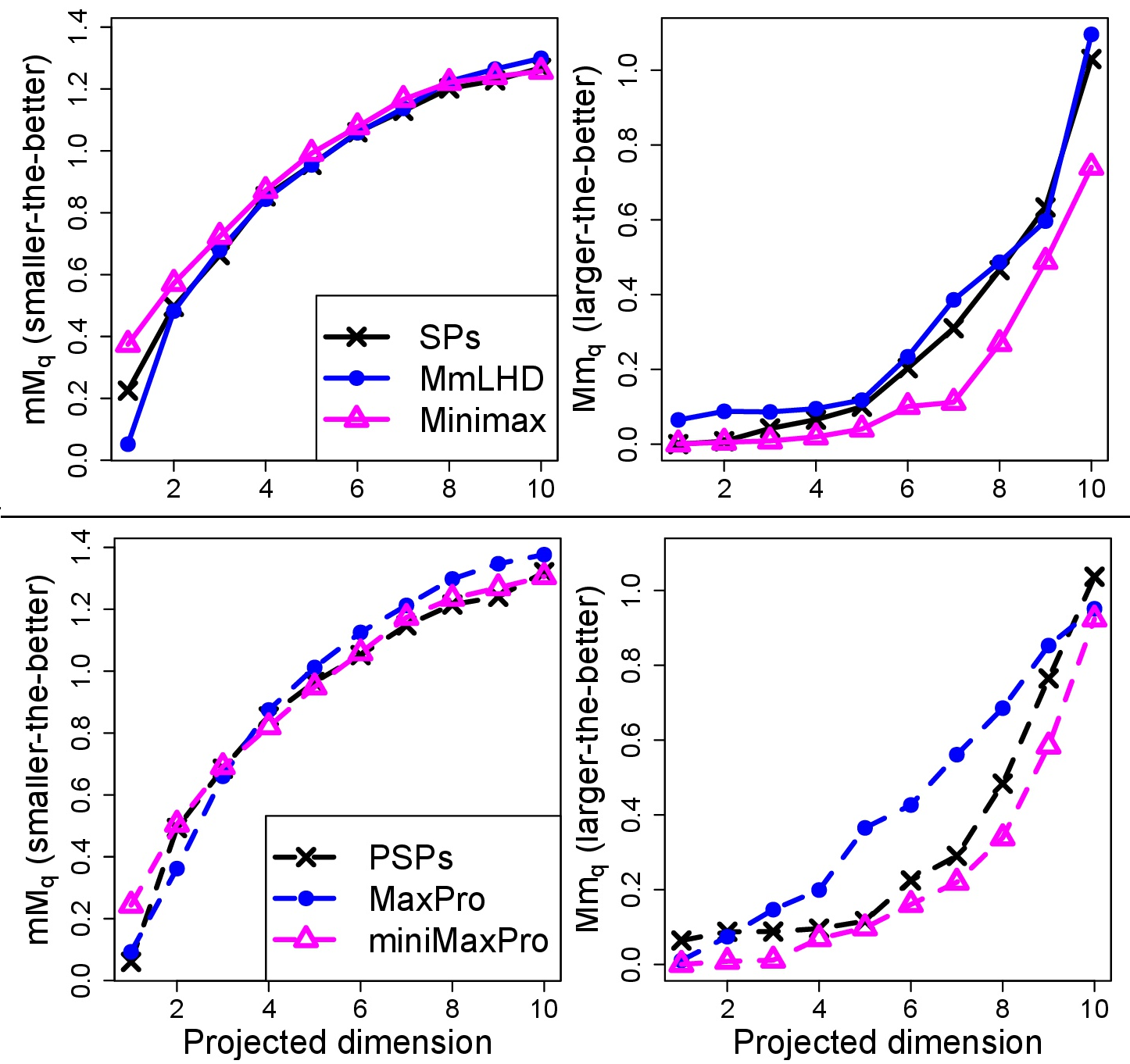}
\caption{(Top) Minimax and maximin index criteria for $n=75$-point SPs on $\mathcal{U}[0,1]^{10}$, MmLHDs and minimax designs. (Bottom) Minimax and maximin index criterion for $n=75$-point PSPs on $\mathcal{U}[0,1]^{10}$, MaxPro and miniMaxPro designs.}
\label{fig:critunif}
\end{figure}

Next, we explore two distance-based metrics for these designs (with the UD replaced by the miniMaxPro design in \cite{MJ2017a}). The first, $\textup{mM}_l(\mathcal{D})$, is the projected minimax index metric \cite{Jea2015}:
\[\textup{mM}_l(\mathcal{D}) = \max_{\emptyset \neq \bm{u} \subseteq [p], |\bm{u}| = l} \; \sup_{\bm{x} \in [0,1]^{\bm{u}}} \left\{ \frac{1}{n} \sum_{i=1}^n \frac{1}{\| \bm{x} - \mathcal{P}_{\bm{u}}\bm{x}_i \|_2^{2l}}\right\}^{-1/(2l)},\]
where $\mathcal{P}_{\bm{u}}$ is the projection operator onto the unit hypercube $[0,1]^{\bm{u}}$. In words, $\textup{mM}_l(\mathcal{D})$ measures the \textit{worst-case minimax performance} of design $\mathcal{D} = \{\bm{x}_i\}_{i=1}^n$, when \textit{projected} onto a subspace of dimension $l$, $l=1, \cdots, p$. The second, $\textup{Mm}_l(\mathcal{D})$, is the projected maximin index metric \citep{Jea2015}:
\[\text{Mm}_l(\mathcal{D}) = \min_{\emptyset \neq \bm{u} \subseteq [p], |\bm{u}|=l} \left\{\frac{1}{{n \choose 2}} \sum_{i=1}^{n-1} \sum_{j=i+1}^n \frac{1}{\|\mathcal{P}_{\bm{u}}\bm{x}_i - \mathcal{P}_{\bm{u}}\bm{x}_j \|_2^{2l}} \right\}^{-1/(2l)}.\]
$\textup{Mm}_l(\mathcal{D})$ measures the \textit{worst-case maximin performance} of design $\mathcal{D}$, when \textit{projected} onto a subspace of dimension $l$. Note that, when $l = p$, both the minimax index criterion $\textup{mM}_p(\mathcal{D})$ and the maximin index criterion $\textup{Mm}_p(\mathcal{D})$ well-approximates the original minimax distance $\textup{mM}(\mathcal{D})$ and maximin distance $\textup{Mm}(\mathcal{D})$, up to constants.

Figure \ref{fig:critunif} plots the minimax and maximin index criteria for each design on $[0,1]^{10}$. Consider first the top two plots comparing $\textup{mM}_l$ and $\textup{Mm}_l$ for SPs, MmLHD and minimax designs -- designs which focus on space-fillingness on the full 10-d space (although MmLHD does have good 1-d projections). For $l=p=10$, the top-left plot shows that minimax designs enjoy the best \textit{minimax} performance in the full space, followed by SPs and MmLHD, while the top-right plot shows that MmLHD enjoys the best \textit{maximin} performance, followed by SPs and minimax designs. Both observations are not too surprising, since the best designs explicitly optimize for their respective metric. What is interesting here is that, for both $\textup{mM}_l$ and $\textup{Mm}_l$, the trade-off between minimax and maximin designs for SPs can be observed numerically. For $\textup{mM}_{10}$, SPs provide slightly worse performance to minimax designs, but much better performance to MmLHDs. Likewise, for $\textup{Mm}_{10}$, SPs perform slightly worse than MmLHDs, but much better than minimax designs. This supports the observation in Section \ref{sec:minmax} that SPs are a compromise between minimax and maximin designs.

Consider next the bottom plots in Figure \ref{fig:critunif}, which compare PSPs, MaxPro and miniMaxPro designs -- designs which account for space-fillingness on all lower-dimensional subspaces. Two observations are of interest here. First, PSPs appear to provide a nice trade-off between MaxPro designs -- which address \textit{maximin} performance on projections -- and miniMaxPro designs -- which address \textit{minimax} performance on projections. Second, we see that the maximin index metric in low dimensions (specifically, $\textup{Mm}_1$) is considerably higher for PSPs than MaxPro. This again shows the PSP correction term $\phi(x;\lambda)$ in Theorem \ref{thm:maxpro} is effective in rectifying the boundary attraction behavior of MaxPro designs in low dimensions.

\subsection{Robustness for Gaussian process modeling}
\label{sec:eff}

\begin{figure}[t]
\centering
\footnotesize
\begin{tabular}{ c " c c " c c " c c " c c}
\specialrule{2pt}{1pt}{1pt}
\textbf{Smoothness} & \multicolumn{4}{c"}{$\Theta = (0,5]^s$} & \multicolumn{4}{c}{$\Theta = (0,20]^s$}\\
\hline
\textbf{\% active} & \multicolumn{2}{c"}{\textit{100\% active}} & \multicolumn{2}{c"}{\textit{40\% active}} & \multicolumn{2}{c"}{\textit{100\% active}} & \multicolumn{2}{c}{\textit{40\% active}} \\
\hline
${(p,n)}$ & $(5,30)$ & $(5,120)$ & $(5,30)$ & $(5,120)$ & $(5,30)$ & $(5,120)$ & $(5,30)$ & $(5,120)$ \\
\specialrule{2pt}{1pt}{1pt}
{SPs} & \ulcrd{1.000} & \ulcrd{0.967} & 0.761 & 0.800 & \ulcrd{0.993} & \ulcrd{0.968} & 0.828 & 0.832\\
{PSPs} & \ull{\cbl{0.938}} & \ull{\cbl{0.927}} & \ulcrd{0.876} & \ulcrd{0.855} & 0.967 & 0.941 & \ulcrd{0.910} & \ulcrd{0.859}\\
{Minimax} & 0.887 & 0.839 & 0.574 & 0.426 & 0.967 & \ull{\cbl{0.960}} & 0.694 & 0.551\\
{miniMaxPro} & 0.915 & 0.906 & 0.765 & 0.629 & \ull{\cbl{0.973}} & 0.945 & 0.793 & 0.784\\
{MmLHD} & 0.917 & 0.873 & \ull{\cbl{0.774}} & \ull{\cbl{0.833}} & 0.960 & 0.926 & \ull{\cbl{0.862}} & 0.832\\
{MaxPro} & 0.823 & 0.819 & 0.772 & 0.775 & 0.900 & 0.871 & 0.822 & \ull{\cbl{0.848}}\\
\specialrule{2pt}{1pt}{1pt}
\end{tabular}
\normalsize
\captionof{table}{$\text{Eff}(\mathcal{D};\Theta)$ for varying function smoothness and proportion of active factors. The double underline $\ull{\ull{\; \cdot \;}}$ highlights the best design (also colored red); the single underline $\ull{\; \cdot \;}$ highlights the second-best design (also colored blue).}
\label{tbl:hard}
\end{figure}

Finally, we inspect the robustness of SPs and PSPs for Gaussian process modeling. Consider first the integrated root-mean-squared error (IRMSE) metric:
\begin{equation}
IRMSE_{\boldsymbol{\theta}}(\mathcal{D}) = \int_{[0,1]^p} RMSE_{\boldsymbol{\theta}}(\bm{x};\mathcal{D}) \; d\bm{x},
\label{eq:irmse}
\end{equation}
with $r_{\boldsymbol{\theta}}(\bm{x},\bm{y})$ being the Gaussian correlation function $\exp\{-\sum_{l=1}^p \theta_l (x_l - y_l)^2\}$. In words, $IRMSE_{\boldsymbol{\theta}}(\mathcal{D})$ quantifies the integrated prediction error from design $\mathcal{D}$ when $f$ is drawn from the Gaussian process $GP\{\mu, r_{\boldsymbol{\theta}}(\cdot,\cdot)\}$. A good design for prediction should therefore minimize such a metric. From our experience, we found the $IRMSE$ criterion to be a more suitable metric than the more popular integrated mean-squared error criterion \citep{Sea1989a}, because the former directly corresponds to confidence intervals for predictions.

Using the IRMSE criterion, we employ the following \textit{efficiency} metric \citep{Sea1989b} to evaluate the robustness of a design for GP modeling:
\begin{equation}
\textup{Eff}(\mathcal{D};\Theta) \equiv \min_{\boldsymbol{\theta} \in \Theta} \frac{IRMSE_{\boldsymbol{\theta}}(\mathcal{D}^*(\boldsymbol{\theta}))}{IRMSE_{\boldsymbol{\theta}}(\mathcal{D})}, \; \mathcal{D}^*(\boldsymbol{\theta}) = \underset{\mathcal{D}}{\textup{argmin}} \; IRMSE_{\boldsymbol{\theta}}(\mathcal{D}).
\label{eq:eff}
\end{equation}
This criterion can be explained in several parts. Here, $\mathcal{D}^*(\boldsymbol{\theta})$ is the optimal design which minimizes IRMSE for fixed correlation parameters $\boldsymbol{\theta}$. Taking the ratio of this optimal error, $IRMSE_{\boldsymbol{\theta}}(\mathcal{D}^*(\boldsymbol{\theta}))$, over the error of the considered design, $IRMSE_{\boldsymbol{\theta}}(\mathcal{D})$, the efficiency $\textup{Eff}(\mathcal{D};\Theta)$ returns the \textit{worst-case} error ratio over a candidate set $\Theta$ of plausible correlation parameters. A larger efficiency therefore suggests a more \textit{robust} design for GP modeling, in situations where correlation parameters are known only to be found in $\Theta$.

With this efficiency criterion, we examine the robustness of these designs in $p=5$ dimensions using $n=30$ and $n=120$-point designs, with $\mathcal{D}^*(\boldsymbol{\theta})$ estimated by the lowest IRMSE amongst the six designs. Letting $s$ be the number of active factors, Table \ref{tbl:hard} summarizes these efficiencies under two classifications: (a) whether all factors are active ($s=p$), or only some are active ($s=0.4p$), and (b) whether the underlying function is smooth ($\Theta = (0,5]^s$) or rough ($\Theta = (0,20]^s$). Here, non-active factors are assigned a correlation parameter of $\theta_l = 0$.

The results in Table \ref{tbl:hard} reveal some interesting insights. First, when the desired function is \textit{smooth} and \textit{all} factors are active, SPs provide the best efficiency of all designs, followed closely by PSPs. When the function is \textit{smooth} but only \textit{some} factors are active, PSPs give the best efficiency, followed by MmLHD. In other words, for smooth functions, SPs and PSPs can provide more robust GP modeling than both minimax-type and maximin-type designs. This is in line with a conjecture in Section \ref{sec:mm} that a trade-off between minimax and maximin can yield better designs than focusing solely on one of these objectives. Second, when the desired function is \textit{rough} and \textit{all} factors are active, SPs again provide the best efficiency over all designs, followed closely by minimax-type designs. When the function is \textit{rough} but only \textit{some} factors are active, PSPs again give the best efficiency, followed by MaxPro designs. These observations were surprising to us initially, since one expects minimax designs to be optimal for rough functions (see \cite{Jea1990}). One explanation is that, in the near-independent GP setting of $\boldsymbol{\theta} \rightarrow \infty$, the predictor $\hat{f}(\cdot)$ in \eqref{eq:pred} converges almost everywhere to the \textit{mean} constant $\mu$, so point sets which are optimal for integration (such as SPs) should perform well in such a setting as well.

\subsubsection{Recommendations}
\begin{figure}[t]
\centering
\small
\begin{tabular}{ c " c c }
\specialrule{2pt}{1pt}{1pt}
\backslashbox{\textbf{\% of active factors}}{\textbf{Smoothness}}& \textit{Smooth} & \textit{Rough}\\
\specialrule{2pt}{1pt}{1pt}
\textit{All active} & SPs & SPs \\
& (PSPs) & (Minimax)\\
\hline
\textit{Some active} & PSPs & PSPs \\
 & (MmLHD) & (MmLHD, MaxPro) \\
\specialrule{2pt}{1pt}{1pt}
\end{tabular}
\captionof{table}{Design recommendations for varying function smoothness and proportion of active factors. Alterative recommendations are bracketed.}
\normalsize
\label{tbl:rec}
\end{figure}

Table \ref{tbl:rec} summarizes the design recommendations made in this section, categorized by the smoothness of the response surface, and whether some or all factors are active. For all categories, SPs and PSPs provide higher efficiencies to existing designs, which confirms the robustness of these designs over a broad range of response surfaces (see Section \ref{sec:robust}). Existing designs do perform well in specific categories, e.g., minimax designs for rough surfaces with all factors active, or MaxPro designs for rough surfaces with some factors active, which explains why such designs have been applied successfully in practice. In the absense of such prior knowledge on function smoothness, the proposed designs can offer improved emulation performance to existing designs over a wider range of problems. We show in the next section how the recommendations in Table \ref{tbl:rec} also hold for common test functions in computer experiments.

\section{Gaussian process modeling of computer experiments}
\label{sec:app}

\begin{figure}[!t]
\centering
\footnotesize
\begin{tabular}{ c " c c c c " c " c c c c " c}
\specialrule{2pt}{1pt}{1pt}
 & \multicolumn{5}{c"}{\textbf{Exponential function}} & \multicolumn{5}{c}{\textbf{Friedman function}}\\
\hline
$n$ & $30$ & $50$ & $70$ & $90$ & Avg. & $50$ & $70$ & $90$ & $110$ & Avg.\\
\specialrule{2pt}{1pt}{1pt}
{SPs} & \ulcrd{1.000} & \ulcrd{1.000} & \ulcrd{1.000} & \ulcbl{0.950} & \ulcrd{0.988} & 0.901 & \ulcrd{1.000} & \ulcrd{1.000} & \ulcbl{0.741} & \ulcbl{0.911}\\
{PSPs} & \ulcbl{0.919} & 0.840 & \ulcbl{0.931} & \ulcrd{1.000} & \ulcbl{0.923} & 0.884 & \ulcbl{0.947} & 0.869 & \ulcrd{1.000} & \ulcrd{0.925}\\
{Minimax} & 0.324 & 0.484 & 0.228 & 0.720 & 0.439 & 0.625 & 0.736 & 0.827 & 0.608 & 0.699\\
{miniMaxPro} & 0.328 & 0.481 & 0.844 & 0.763 & 0.604 & 0.789 & 0.945 & 0.847 & 0.640 & 0.805\\
{MmLHD} & 0.718 & \ulcbl{0.893} & 0.815 & 0.762 & 0.797 & \ulcbl{0.981} & 0.845 & 0.848 & 0.692 & 0.842\\
{MaxPro} & 0.623 & 0.812 & 0.562 & 0.876 & 0.718 & \ulcrd{1.000} & 0.868 & \ulcbl{0.943} & 0.608 & 0.855\\
\specialrule{2pt}{1pt}{1pt}
\end{tabular}
\begin{tabular}{ c " c c c c " c " c c c c " c}
\specialrule{2pt}{1pt}{1pt}
& \multicolumn{5}{c"}{\textbf{8-dimensional function}} & \multicolumn{5}{c}{\textbf{Wing weight function}}\\
\hline
$n$ & $80$ & $100$ & $120$ & $140$ & Avg. & $100$ & $110$ & $120$ & $130$ & Avg.\\
\specialrule{2pt}{1pt}{1pt}
{SPs} & \ulcbl{0.978} & \ulcrd{1.000} & \ulcbl{0.892} & \ulcbl{0.840} & \ulcbl{0.928} & 0.865 & \ulcbl{0.799} & \ulcrd{1.000} & \ulcrd{1.000} & \ulcbl{0.916}\\
{PSPs} & \ulcrd{1.000} & \ulcbl{0.809} & \ulcrd{1.000} & \ulcrd{1.000} & \ulcrd{0.952} & \ulcrd{1.000} & \ulcrd{1.000} & \ulcbl{0.955} & \ulcbl{0.931} & \ulcrd{0.972}\\
{Minimax} & 0.100 & 0.327 & 0.343 & 0.290 & 0.265 & 0.512 & 0.470 & 0.449 & 0.441 & 0.470\\
{miniMaxPro} & 0.529 & 0.572 & 0.643 & 0.740 & 0.621 & 0.720 & 0.751 & 0.907 & 0.799 & 0.794\\
{MmLHD} & 0.630 & 0.681 & 0.850 & 0.821 & 0.746 & \ulcbl{0.895} & 0.694 & 0.929 & 0.890 & 0.852\\
{MaxPro} & 0.794 & 0.667 & 0.758 & 0.796 & 0.754 & 0.847 & 0.793 & 0.822 & 0.721 & 0.796\\
\specialrule{2pt}{1pt}{1pt}
\end{tabular}
\normalsize
\captionof{table}{$\widetilde{\textup{Eff}}(\mathcal{D})$ for various functions and design sizes $n$. The double underline $\ull{\ull{\; \cdot \;}}$ highlights the best design (also colored red); the single underline $\ull{\; \cdot \;}$ highlights the second-best design (also colored blue).}
\label{tbl:hardfn}
\end{figure}

We now investigate the emulation performance of SPs and PSPs for four computer experiment functions in the literature. The first is the exponential function in \cite{DP2010}, which is asymptotic near design boundaries. The second is the test function in \cite{Fea1983}, composed of an interaction term and several 1-d additive terms. The third is a computer experiment function in \cite{DP2010}, highly curved in some variables but much less in others. The last is the wing weight function in \cite{Fea2008}, which models the weight of a light aircraft wing.

Here, we use the GP predictor $\hat{f}(\cdot)$ in \eqref{eq:pred} for emulating the response surface $f$ (as implemented in the \textsf{R} package \texttt{GPfit}; \cite{Mea2015}), with $r_{\boldsymbol{\theta}}(\cdot,\cdot)$ taken to be the Gaussian correlation. To measure the quality of a design, we use the efficiency metric:
\[\widetilde{\textup{Eff}}(\mathcal{D}) = \frac{\widetilde{IRMSE}(\mathcal{D^*})}{\widetilde{IRMSE}(\mathcal{D})}, \quad \widetilde{IRMSE}(\mathcal{D}) = \int_{[0,1]^p} |f(\bm{x}) - \hat{f}(\bm{x})| \; d\bm{x},\]
where $\mathcal{D}^*$ is the optimal design (of the six considered) minimizing $\widetilde{IRMSE}$.

For each design, Table \ref{tbl:hardfn} summarizes their efficiencies $\widetilde{\textup{Eff}}(\mathcal{D})$ for several design sizes $n$, along with their average efficiency over all $n$. We make two important remarks here. First, SPs and PSPs enjoy improved emulation performance to the four existing designs for nearly all design sizes $n$, and enjoy considerably better efficiencies on average. This again supports the theoretical argument in Section \ref{sec:robust} that the proposed designs are \textit{robust} for modeling a wide range of computer experiment problems. The results in Table \ref{tbl:hardfn} also corroborate the design recommendations made in Table \ref{tbl:rec}. For example, SPs provide the best emulation of the exponential function, which is smooth and active in all factors, a result consistent with Table \ref{tbl:rec}. Likewise, PSPs provide excellent emulation of the wing weight function, which is smooth and active in some factors, an observation also in line with Table \ref{tbl:rec}.

Second, Table \ref{tbl:hardfn} shows that the adjustments made by SPs and PSPs over existing designs indeed results in improved emulation performance. Comparing SPs with minimax designs and MmLHDs, we see that by \textit{trading-off} between the minimax and maximin objectives, SPs offer better predictive performance than designs which focus solely on minimax or maximin. Similarly, comparing PSPs with MaxPro designs, we see that by correcting the boundary attraction behavior of MaxPro design points, the resulting PSPs yield considerable emulation improvement over the latter design. The designs proposed in this paper therefore offer a unified framework which corrects certain deficiencies of existing designs, and as such, provide improved emulation performance over existing methods.

\section{Conclusion}
\label{sec:concl}
In this paper, we presented a theoretical and practical argument for support points and projected support points as robust and effective designs for computer experiments. Theoretically, we showed that the proposed designs are (a) robust for Gaussian process modeling, (b) account for the three fundamental design principles of effect sparsity, effect hierarchy and effect heredity, and (c) yield improvements over popular computer experiment designs in the literature. Practically, we illustrated several computational advantages of support points and projected support points with regards to design optimization, then showed using a suite of numerical experiments that the proposed designs can yield considerably improved emulation performance over existing designs.

Looking forward, there are several exciting directions for future research. While the two-part minimax and one-part maximin trade-off for support points (Proposition \ref{thm:dist}) indeed provides improved designs, it would be interesting to see whether a different trade-off (perhaps depending on design size $n$ or dimension $p$) can yield further improvements for emulation. Also, given the excellent performance of the proposed designs in the uniform setting, it may be worthwhile exploring the use of SPs and PSPs on \textit{non-uniform} distributions as designs, particularly in scenarios where a \textit{weighted} prediction accuracy is desired over the design space.

The algorithms in this paper for generating SPs and PSPs are provided in the \textsc{R} package \texttt{support} \citep{Mak2017}, which will be made available on CRAN.

\textbf{Acknowledgements}: This research is supported by a U.S. National Science Foundation grant DMS-1712642 and a U.S. Army Research Office grant W911NF-17-1-0007.

\vspace{-0.25cm}

\footnotesize{\bibliography{references}}
\end{document}